\documentclass[runningheads]{llncs}

\usepackage[utf8]{inputenc}
\usepackage[UKenglish]{babel}

\usepackage[paperheight=235mm, paperwidth=155mm,textwidth=12.2cm,textheight=19.3cm]{geometry}

\usepackage{bbding}
\usepackage{adjustbox}

\usepackage{amsthm,amsmath,amssymb,mathtools,thmtools}
\usepackage{xspace}
\usepackage{xcolor}
\usepackage{graphicx}
\usepackage{url}
\usepackage{stmaryrd}
\RequirePackage[bookmarks,colorlinks=true]{hyperref}

\usepackage[capitalise]{cleveref}
\usepackage{arydshln} 
\usepackage{cancel}
\usepackage{booktabs}
\usepackage{multirow}
\usepackage{tikz}
\usetikzlibrary{arrows,automata,decorations.pathmorphing}
\usetikzlibrary{positioning}
\usetikzlibrary{shapes,shapes.geometric,fit,calc,automata}
\usepackage{enumitem}

\newcommand{\Func}[1]{\mathsf{#1}}
\newcommand{\tuple}[1]{\langle #1 \rangle}

\newcommand{\Nat}{\ensuremath{\mathbb{N}}\xspace}

\newcommand{\true}{{\ensuremath{\mathbf{true}}}}
\newcommand{\false}{{\ensuremath{\mathbf{false}}}}
\newcommand{\AP}{AP}
\newcommand{\F}{{\ensuremath{\mathbf{F}}}}
\newcommand{\G}{{\ensuremath{\mathbf{G}}}}
\newcommand{\X}{{\ensuremath{\mathbf{X}}}}
\newcommand{\U}{{\ensuremath{\mathbf{U}}}}
\newcommand{\R}{{\ensuremath{\mathbf{R}}}}
\newcommand{\depth}[1]{\mathsf{depth}(#1)}

\newcommand{\Enter}[1]{\Func{Enter}(#1)}
\newcommand{\Stay}[1]{\Func{Stay}(#1)}
\newcommand{\Leave}[1]{\Func{Leave}(#1)}
\newcommand{\InitState}{\text{\Large $ \iota$}}

\newcommand{\A}{{\cal A}}

\newcommand{\D}{{\cal D}}

\newcommand{\Inf}{\it inf}

\newcommand{\FIN}{{\mathit{Fin}}}

\newcommand{\ZB}{,\mathbf{0}}
\newcommand{\OB}{,\mathbf{1}}
\newcommand{\OrState}{\!\!\!\small${\lor}$\!\!\!}
\newcommand{\AndState}{\!\!\!\small${\land}$\!\!\!}

\newcommand{\Subject}[1]{\paragraph*{#1.}}
\definecolor{Green}{rgb}{0.13, 0.55, 0.13}
\renewcommand{\phi}{\varphi}
\newcommand{\satisfies}{\models}
\newcommand{\Intui}[1]{\textcolor{blue}{\text{#1}}}

\newcommand{\LF}{\mathsf{L}}
\newcommand{\DF}{\mathsf{D}}

\makeatletter

\newcommand*{\da@rightarrow}{\mathchar"0\hexnumber@\symAMSa 4B }
\newcommand*{\da@leftarrow}{\mathchar"0\hexnumber@\symAMSa 4C }

\newcommand*{\xmyrightarrow}[3][~~~]{
	\mathrel{%
		\mathpalette{\da@xarrow{#1}{#2}{}{\hspace{-1.2pt}\da@rightarrow}{\,}{}{}{#3}}{}%
	}%
}

\newcommand*{\xmyleftarrow}[3][~~~]{%
	\mathrel{%
		\mathpalette{\da@xarrow{#1}{#2}{}\da@leftarrow{\,}{}{}{#3}}{}%
	}%
}

\newcommand*{\xdashrightarrow}[2][~~~]{%
	\mathrel{%
		\mathpalette{\da@xarrow{#1}{#2}{}{\da@rightarrow}{\,}{}{}{\dabar@}}{}%
	}%
}
\newcommand{\xdashleftarrow}[2][~~~]{%
	\mathrel{%
		\mathpalette{\da@xarrow{#1}{#2}\da@leftarrow{}{}{\,}{}{\dabar@}}{}%
	}%
}
\newcommand*{\da@xarrow}[8]{%
	\sbox0{$\ifx#7\scriptstyle\scriptscriptstyle\else\scriptstyle\fi#5#1#6\m@th$}%
	\sbox2{$\ifx#7\scriptstyle\scriptscriptstyle\else\scriptstyle\fi#5#2#6\m@th$}%
	\sbox4{$#7#8\m@th$}%
	\dimen@=\wd0 %
	\ifdim\wd2 >\dimen@
	\dimen@=\wd2 %
	\fi
	\count@=2 %
	\def\da@bars{#8#8}%
	\@whiledim\count@\wd4<\dimen@\do{%
		\advance\count@\@ne
		\expandafter\def\expandafter\da@bars\expandafter{%
			\da@bars
			#8
		}%
	}%
	\mathrel{#3}%
	\mathrel{%
		\mathop{\da@bars}\limits
		\ifx\\#1\\%
		\else
		_{\copy0}%
		\fi
		\ifx\\#2\\%
		\else
		^{\copy2}%
		\fi
	}%
	\mathrel{#4}\!\!%
}
\makeatother

\newcommand*\Trim[1]{\adjustbox{clip,trim={.115\width} 0pt {.13\width} 0pt }{\ensuremath{#1}}}
\newcommand{\ArrowSymbol}{\hspace{-0.3pt}\Trim{\thicksim}\hspace{0.2pt}}

\newcommand{\Strong}{}
\newcommand{\Weak}{\text{\tiny weak}}

\newcommand{\ReachMulti}[6]{\ensuremath{#1\xmyrightarrow[\bcancel{#2(#3)}]{#6}{\ArrowSymbol} #4\left(#5\right)}}
\newcommand{\Reach}[5]{\ReachMulti{#1}{#2}{#3}{#4}{#5}{\Strong}}
\newcommand{\WeakReach}[5]{\ReachMulti{#1}{#2}{#3}{#4}{#5}{\Weak}}

\newcommand{\FromSecond}{>0}
\newcommand{\StayReachMulti}[6]{\ensuremath{#1\xrightarrow[\bcancel{#2\left(#3\right)}]{#6} #4\left(#5\right)}}
\newcommand{\StayReach}[5]{\StayReachMulti{#1}{#2}{#3}{#4}{#5}{\Strong}}
\newcommand{\WeakStayReach}[5]{\StayReachMulti{#1}{#2}{#3}{#4}{#5}{\Weak}}
\newcommand{\StayReachNonEmpty}[5]{\StayReachMulti{#1}{#2}{#3}{#4}{#5}{\FromSecond}}
\newcommand{\WeakStayReachNonEmpty}[5]{\StayReachMulti{#1}{#2}{#3}{#4}{#5}{\Weak,\FromSecond}}

\newcommand{\LeaveReachMulti}[6]{\ensuremath{#1\xdashrightarrow[\bcancel{#2\left(#3\right)}]{#6} #4\left(#5\right)}}
\newcommand{\LeaveReach}[5]{\LeaveReachMulti{#1}{#2}{#3}{#4}{#5}{\Strong}}

\newcommand{\ReachTT}[2]{\ensuremath{#1\xmyrightarrow{}{\ArrowSymbol} #2}}
\newcommand{\NonEmptyReachTT}[2]{\ensuremath{#1\xmyrightarrow{>0}{\ArrowSymbol} #2}}
\newcommand{\smallerReach}[1]{{\footnotesize #1}}

\newenvironment{changemargin}[2]{%
	\begin{list}{}{%
			\setlength{\topsep}{0pt}%
			\setlength{\leftmargin}{#1}%
			\setlength{\rightmargin}{#2}%
			\setlength{\listparindent}{\parindent}%
			\setlength{\itemindent}{\parindent}%
			\setlength{\parsep}{\parskip}%
		}%
		\item[]}{\end{list}}

\makeatletter
\def\@citecolor{blue}%
\def\@urlcolor{blue}%
\def\@linkcolor{blue}%

\def\orcidID#1{\smash{\href{http://orcid.org/#1}{\protect\raisebox{-1.25pt}{\protect\includegraphics{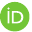}}}}}
\makeatother

\begin{document}
\title{On the Translation of Automata\\to Linear Temporal Logic\thanks{This is the full version of a chapter with the same title that appears in the FoSSaCS 2022 conference proceedings \cite{DBLP:conf/fossacs/BokerLS22}.}}

\author{Udi~Boker\inst{1} \orcidID{0000-0003-4322-8892} \and
Karoliina~Lehtinen\inst{2} \orcidID{0000-0003-1171-8790} \and
Salomon~Sickert\inst{3}\thanks{Salomon Sickert is supported by the Deutsche Forschungsgemeinschaft (DFG) under project number 436811179.} \orcidID{0000-0002-0280-8981}
}
\authorrunning{U. Boker, K. Lehtinen, S. Sickert}
\titlerunning{On the Translation of Automata to Linear Temporal Logic}
\institute{Reichman University, Herzliya, Israel\\
\and
CNRS, Aix-Marseille University and University of Toulon, LIS, Marseille, France
\and
The Hebrew University, Jerusalem, Israel\\
}

\maketitle              
\begin{abstract}
While the complexity of translating future linear temporal logic (LTL)
into automata on infinite words is well-understood, the size increase involved in turning automata back to LTL is not. In particular, there is no known elementary bound on the complexity of translating deterministic $\omega$-regular automata to LTL.

Our first contribution consists of tight bounds for LTL over a unary alphabet: alternating, nondeterministic and deterministic automata can be exactly exponentially, quadratically and linearly more succinct, respectively, than any equivalent LTL formula.
Our main contribution consists of a translation of general counter-free deterministic $\omega$-regular automata into LTL formulas of double exponential temporal-nesting depth and triple exponential length, using an intermediate Krohn-Rhodes cascade decomposition of the automaton. To our knowledge, this is the first elementary bound on this translation. Furthermore, our translation preserves the acceptance condition of the automaton in the sense that it turns a looping, weak, B\"uchi, coB\"uchi or Muller automaton into a formula that belongs to the matching class of the syntactic future hierarchy. In particular, it can be used to translate an LTL formula recognising a safety language to a formula belonging to the safety fragment of LTL (over both finite and infinite words). 
\end{abstract}

\section{Introduction}

Linear Temporal Logic with only future temporal operators (from here on LTL) and $\omega$-regular automata, whether deterministic, nondeterministic or alternating, are both well-established formalisms to describe properties of infinite-word languages.
LTL is popular in formal verification and synthesis due to its simple syntax and semantics.
Yet, while properties might be convenient to define in LTL, most verification and synthesis algorithms eventually compile LTL formulas into $\omega$-regular automata.
The expressiveness of both these key formalisms, as well as translations from LTL to automata of various types, are well understood. Here, we consider the converse translations, which, in comparison, have received less attention: up till now, no elementary upper bound on the size blow-up of going from automata to LTL was known.

Regarding expressive power, deterministic Muller automata, nondeterministic B\"uchi automata, and weak alternating automata recognise all $\omega$-regular languages~\cite{KV01,Tho90}. LTL-definable languages (surveyed in~\cite{DG08}) are a strict subset thereof, also defined by first-order logic, star-free regular expressions, aperiodic monoids, counter-free automata, and very weak alternating automata.
As for succinctness, nondeterministic and alternating automata can be exponentially and double-exponentially more succinct than deterministic automata, respectively. Determinisation in particular has precise bounds~\cite{Mic88,Saf89,Lod99,Sch09,CZ09,Bok18}.

The succinctness of various representations of LTL-definable languages is less clear: effective translations between the different models are far from straightforward, and their complexity is sometimes uncertain. In particular, to the best of our knowledge, up to now there has been no elementary bound even on the translation of deterministic counter-free automata, arguably the simplest automata model for this class of languages, into LTL formulas. 
(Considering LTL with both future and past temporal operators, there is a double-exponential upper bound on the length of the formula \cite{MP90}\footnote{See \cref{rem:Cascade} on whether the upper bound in \cite{MP90} is single or double exponential.}.) The complexity of obtaining a deterministic counter-free automaton from a nondeterministic one is also, to the best of our knowledge, open.

We study the complexity of translating automata to LTL (equivalently, to very weak alternating automata), considering formula length, size, and nesting depth of temporal operators.

We begin (Section~\ref{sec:Unary}), as a warm-up, with the unary alphabet case on finite words. 
We show that the size-blow up involved in translating deterministic, non-deterministic and alternating automata to LTL, when possible, is linear, quadratic and exponential, respectively, and these bounds are tight. In contrast, going from LTL to alternating, nondeterministic and deterministic automata is linear, exponential and double-exponential, respectively~\cite{MSS88,VW86,KR10}. 

The case of non-unary alphabets is much more difficult. We provide a translation of counter-free deterministic $\omega$-regular automata (with any acceptance condition) into LTL formulas with double exponential depth and triple exponential length. 
Our translation uses an intermediate Krohn-Rhodes \emph{reset cascade decomposition} (wreath product) of deterministic automata, which is a deterministic automaton built from simple components.

Our main technical contribution consists of a translation of a reset cascade into an LTL formula of depth linear and length singly exponential in the number of cascade configurations. Combining this with Eilenberg's Holonomy translation of a semigroup into a cascade \cite[Corollary II.7.2]{Eil76} and Pnueli and Maler's adaptation of it to automata \cite[Theorem 3]{MP90} (see \cref{rem:Cascade}), we obtain a translation of counter-free deterministic $\omega$-regular automata into LTL formulas of double exponential depth and triple exponential length. Our construction preserves the acceptance condition of the automaton in the sense that it turns a B\"uchi-looping, coB\"uchi-looping, weak, B\"uchi or coB\"uchi automaton into a formula that belongs to the matching class of the syntactic future hierarchy (see \cref{def:future_hierarchy} and \cite{ChangMP92}).
	 
\subsubsection*{Related work}
\Subject{Finite words} While LTL is usually interpreted over infinite words, it also admits finite-word semantics that coincide with the finite word version of the other equivalent formalisms. The equivalence between FO and star-free languages on finite words is due to McNaughton and Papert~\cite{MP71}. Cohen, Perrin and Pin~\cite{CPP93} used the Krohn-Rhodes decomposition to characterise the expressive power of LTL with only $\X$ and $\F$ (eventually), but do not provide bounds on the size trade-off between the different models. Wilke~\cite{Wil99} gives a double-exponential translation from counter-free DFA to LTL. More recently, Boja\'nczyk provided an algebraically flavoured adaptation of Wilke's proof~\cite[Section 2.2.2]{Boj20}.

\Subject{Infinite words} With substantial effort over several decades, the above techniques have been extended to infinite words using intricate tools with opaque complexities. Ladner~\cite{Lad77} and Thomas~\cite{Tho79,Tho81} for example extended the equivalence of star-free regular expressions and FO to infinite words, while the $\omega$-extension of the equivalence with aperiodic languages is due to Perrin~\cite{Per84}. The correspondence with LTL is due to Kamp~\cite{Kam68} and Gabbay, Pnueli, Shelah and Stavi~\cite{GPSS80}.
Diekert and Gastin's survey~\cite{DG08} provides an algebraic translation into LTL via $\omega$-monoids while Cohen-Chesnot gives a direct algebraic proof of the equivalence of star-free $\omega$-regular expressions and LTL~\cite{Coh91}.
Wilke takes an automata-theoretic approach, using backward deterministic automata~\cite{Wil16,Wil18}.
However, none of the above address the complexity of the transformations. 
Zuck's dissertation~\cite{XXXX:phd/Zuck86} gives a translation of star-free regular expressions into LTL, with at least non-elementary complexity. Subsequently, Chang, Mana and Pneuli \cite{ChangMP92} use Zuck's results to show that the levels of their hierarchy of future temporal properties coincide with syntactic fragments of LTL. Sickert and Esparza \cite{SickertE20} gave an exponential translation of any LTL formula into level $\Delta_2$ of this hierarchy.

\section{Preliminaries}

\Subject{Languages}

An alphabet $\Sigma$, of size $|\Sigma|$, is a finite set of letters. $\Sigma^*, \Sigma^+$, and $\Sigma^\omega$ denote the sets of finite, nonempty finite, and infinite words over $\Sigma$, respectively. A language of finite or infinite words is a subset of $\Sigma^*$ or $\Sigma^\omega$, respectively.
We write $[i..j]$ and $[i..j)$, with integers $i\leq j$, for the sets $\{i, i+1, \ldots, j\}$ and $\{i, i+1, \ldots, j-1\}$, respectively.
For a word $w = \sigma_0 \cdot \sigma_1  \cdots$, we write $|w|$ for its length ($\infty$ if $w$ is infinite), $w[i]$ for $\sigma_i$, $w_{[i..j]}$ and $w_{[i..j)}$ for its corresponding infixes ($w_{[i..i)}$ is the empty word), and $w_{[i..]}$ for its (finite or infinite) suffix $\sigma_i \cdot \sigma_{i+1} \cdots$.

\Subject{Linear Temporal Logic (LTL)}

Let $\AP$ be a finite set of atomic propositions. LTL formulas are constructed from the constant $\true$, atomic propositions $a \in \AP$, the connectives $\neg$ (negation) and $\wedge$ (and), and the temporal operators $\U$ (until) and $\X$ (next).
Their semantics are given by a satisfiability relation $\models$ between finite or infinite words $w  \in (2^{\AP})^+ \cup (2^{\AP})^\omega$, and a formula $\varphi$ inductively as follows:
{\arraycolsep=3pt%
\[\begin{array}[t]{ll}
w \models \true & \hspace{14.1em} w \models a  \hspace{2.64em}\textit{iff~~} a \in w[0] \\
w \models \neg \varphi &  \textit{iff~~}  w \not \models \varphi \hspace{9.6em} w \models \varphi \wedge \psi   \textit{~~iff~~}  w \models \varphi \textit{ and } w \models \psi   \\
w \models \X \varphi   & \textit{iff~~}          |w| > 1 \textit{ and } w_{[1..]} \models \varphi \\
w \models \varphi \U \psi & \textit{iff~~}  \exists i \in [0..|w|). ~ w_{[i..]} \models \psi \textit{ and } \forall j \in [0..i). ~ w_{[j..]} \models \varphi
\end{array}\]}

\noindent We also use the common shortcuts  $\false \coloneqq \neg \true$, $\varphi \lor \psi \coloneqq \neg ((\neg \varphi) \land (\neg \psi))$, $\F \varphi \coloneqq \true \U \varphi$, $\G \varphi \coloneqq \neg \F \neg \varphi$, and $\psi_1 \R \psi_2 \coloneqq \neg (\neg \psi_1) \U (\neg \psi_2)$.
The language of finite words of $\varphi$ is $L^{<\omega}(\varphi) \coloneqq \{ w \in (2^{\AP})^+ \mid w \models \varphi \}$, and the language of infinite words is $L(\varphi) \coloneqq \{ w \in (2^{\AP})^\omega \mid w \models \varphi\}$. Note that we omit the ``$<\omega$'' superscript if it is clear from the context which set is used. 
The \emph{length} $|\varphi|$ of $\varphi$ is the number of nodes in its syntax tree,
the \emph{size} of $\phi$ is the number of nodes in a DAG representing this syntax tree, and its \emph{temporal nesting depth}, denoted by $\depth{\varphi}$, is defined by: $\depth{\true}=0$; $\depth{a}=0$ for an atomic proposition $a \in \AP$; $\depth{\neg \psi} = \depth{\psi}$; $\depth{\psi_1 \land \psi_2} = \max(\depth{\psi_1}, \depth{\psi_2})$; $\depth{\X \psi} = \depth{\psi}+1$; and $\depth{\psi_1 \U \psi_2} = \max(\depth{\psi_1}, \depth{\psi_2}) + 1$.
Chang, Manna, and Pnueli define in \cite{ChangMP92} a syntactic hierarchy for LTL formulas (over infinite words):

\begin{definition}[LTL Syntactic future hierarchy {\cite{ChangMP92}}
	\footnote{This extends \cite{CernaP03,SickertE20} with negation, which can be removed via negation normal form.}]
\label{def:future_hierarchy}
\begin{itemize}
	\item $\Sigma_0 = \Pi_0 = \Delta_0$ is the least set  containing all atomic propositions and their negations, and is closed under the application of conjunction and disjunction.
	\item $\Sigma_{i+1}$ is the least set containing $\Pi_i$ and negated formulas of $\Pi_{i+1}$  closed under the application of conjunction, disjunction, and the $\X$ and $\U$ operators.
	\item $\Pi_{i+1}$ is the least set containing $\Sigma_i$ and negated formulas of $\Sigma_{i+1}$  closed under the application of conjunction, disjunction, and the $\X$ and $\R$ operators.
	\item $\Delta_{i+1}$ is the least set containing $\Sigma_{i+1}$ and $\Pi_{i+1}$ that is closed under the application of conjunction, disjunction, and negation.
\end{itemize}
\end{definition}

\noindent $\Sigma_1$ is referred to as \emph{syntactic co-safety} formulas, $\Pi_1$ as \emph{syntactic safety} formulas.

\Subject{Automata}

A \emph{deterministic semiautomaton} is a tuple $\D=(\Sigma,Q,\delta)$, where $\Sigma$ is an alphabet; $Q$ is a finite nonempty set of states; and $\delta\colon Q\times \Sigma \to Q$ is a transition function and we extend it to finite words in the usual way. A \emph{path} of $\D$ on a word $w=\sigma_0 \cdot \sigma_1 \cdots$ is a sequence of states $q_0, q_1, \ldots$, such that for every $i<|w|$, we have $\delta(q_i, \sigma_i)=q_{i+1}$.

It is a \emph{reset} semiautomaton if for every letter $\sigma\in\Sigma$, either i) for every state $q\in Q$ we have $\delta(q,\sigma)=q$, or ii) there exists a state $q'\in Q$, such that for every state $q\in Q$ we have $\delta(q,\sigma)=q'$.

It is \emph{counter free} if for every state $q\in Q$, finite word $u\in\Sigma^+$, and number $n\in\Nat\setminus\{0\}$, there is a self loop of $q$ on $u^n$ if{}f there is a self loop of $q$ on $u$.

A \emph{deterministic automaton} is a tuple $\D=(\Sigma,Q,\iota,\delta,\alpha)$, where $(\Sigma,Q,\delta)$ is a deterministic semiautomaton, $\iota \in Q$ is an initial state; and $\alpha$ is some acceptance condition, as detailed below. A run of $\D$ on a word $w$ is a path of $\D$ on $w$ that starts in $\iota$.
 It is a reset or counter-free automaton if its semiautomaton is.

 The \emph{acceptance condition} of an automaton on finite words is a set $F\subseteq Q$; a run is accepting if it ends in a state $q\in F$.
 The \emph{acceptance condition} of an \emph{$\omega$-regular automaton}, on infinite words, is defined with respect to the set $\Inf(r)$ of states visited infinitely often along a run $r$. We define below several acceptance conditions that we use in the sequel; for other conditions, see, for example, \cite{Bok18}.

The \emph{Muller} condition is a set $\alpha=\{M_1, \ldots, M_k\}$ of sets $M_i\subseteq Q$ of states, and a run $r$ is accepting if there exists a set $M_i$, such that $M_i=\Inf(r)$. 
The \emph{Rabin} condition is a set $\alpha=\{(G_1, B_1), \ldots, (G_k, B_k)\}$ of pairs of sets of states, and $r$ is accepting if there exists a pair $(G_i, B_i)$, such that $G_i\cap\Inf(r)\neq\emptyset$ and $B_i\cap\Inf(r)=\emptyset$.
 The \emph{B\"uchi} (resp.\ \emph{coB\"uchi}) condition is a set $\alpha\subseteq Q$ of states, and $r$ is accepting if $\alpha\cap\Inf(r)\neq\emptyset$ (resp.\ $\alpha\cap\Inf(r)=\emptyset$).  
 A \emph{weak} automaton is a B\"uchi automaton, in which every strongly connected component (SCC) contains only states in $\alpha$ or only states out of $\alpha$.
 A \emph{looping} automaton is a B\"uchi or coB\"uchi automaton, where all states are in $\alpha$, except for a single sink state.
 
Deterministic automata of the above types correspond to the hierarchy of temporal properties \cite{MannaP89}:
Looping-B\"uchi, looping-coB\"uchi, weak, B\"uchi, coB\"uchi, and Rabin/Muller deterministic automata define respectively safety, guarantee (co-safety), obligation, recurrence, persistence, and reactivity languages. If the language is also LTL-definable, then there exists an equivalent LTL formula in $\Pi_1$, $\Sigma_1$, $\Delta_1$, $\Pi_2$, $\Sigma_2$, and $\Delta_2$, respectively~\cite{ChangMP92}. 
Every deterministic $\omega$-regular automaton is equivalent to deterministic Muller and Rabin automata, where the Muller (but not always Rabin) one can be defined on the same semiautomaton.

\emph{Nondeterministic} and \emph{alternating} automata (to which we only refer in \cref{sec:Unary}, on finite words over a unary alphabet) extend deterministic automata by having a transition function $\delta\colon Q\times \Sigma \to 2^Q$ and $\delta\colon Q\times \Sigma \to$ (positive Boolean formulas over $Q$), respectively. (See, for example, \cite{CKS81} for formal definitions.)

\section{Unary Alphabet}\label{sec:Unary}

Kupferman, Ta-Shma and Vardi~\cite{KTV99} compared the succinctness of different automata models when \textit{counting}, that is, recognising the singleton language $\{a^k\}$ for some $k$ over the singleton alphabet $\{a\}$. For the succinctness gap between automata and LTL, we study the task of recognising arbitrary languages over the unary alphabet, which can be seen as sets of integers, rather than a single integer.

For a unary alphabet, since there is only one infinite word, only languages on finite words are interesting. We thus consider  LTL formulas over (no) atomic propositions $\AP = \emptyset$, and automata on finite unary words over the corresponding alphabet $\Sigma = 2^{\AP} = \{\emptyset\}$, where we use the shorthand $a = \emptyset$.
The size of a deterministic automaton is the number of its states, of a nondeterministic automaton the number of its transitions, and of an alternating automaton the number of subformulas in its transition function.

{We show }that the size blow-up involved in translating deterministic, nondeterministic, and alternating automata to LTL, when possible, is linear, quadratic, and exponential, respectively.\\

In our analysis, we shall use the following folklore theorem, which extends Wolper's Theorem \cite{Wol83}. 
The proof is given in \Cref{appendix:unary}.

\begin{restatable}[Extended Wolper's theorem, Folklore]{proposition}{extendedWolper}
	\label{cl:Nesting}
	Consider an LTL formula $\phi$ with $\depth{\phi}=n$ over the atomic propositions $\AP$, and let $\Sigma=2^{\AP}$.
	Then for every  words $u\in\Sigma^*$, $v\in\Sigma^+$ and $t\in\Sigma^\omega$, and numbers $i,j>n$, $\phi$ has the same truth value on the words $(u v^i t)$ and $(u v^j t)$.
\end{restatable}

We use this to establish that unary LTL describes only finite and co-finite properties, and that there is a tight relation between the depth of LTL formulas and the length of words above which they are all in or all out of the language.

\begin{restatable}{proposition}{FiniteCofinite}
\label{cl:finite-cofinite}
Given an LTL formula $\phi$ with $\depth{\phi}=n$ on finite words over the unary alphabet $\{a\}$, $a^i \in L(\varphi)$ for all $i>n$ or $a^i \notin L(\phi)$ for all $i>n$.
\end{restatable}

\begin{restatable}{proposition}{LTLforShortLang}
\label{cl:LTLforShortLang}
	Consider a language $L\subseteq\{a\}^+$ that agrees on all words of length over $n$, that is, has the same truth value on all such words.
	Then there is an LTL formula of size in $O(n)$ with language $L$.
\end{restatable}

We now establish the trade-off between LTL and alternating automata (AFA) over unary alphabets. AFA are closed under (linear) complementation, so we use a pumping argument to bound the length after which all words have the same truth value, giving an upper bound on the LTL formula.

\begin{restatable}{lemma}{UnaryAltUpper}\label{cl:UnaryAltUpper}
Every alternating automaton with $n$ states that recognises an LTL-expressible language $L\subseteq \{a\}^+$ is equivalent to an LTL formula of size in $O(2^{n})$.
\end{restatable}

We show next that this upper bound is tight. Consider the language $\{a^{2^{n-1}}\}$, which, according to~\cref{cl:finite-cofinite}, is only recognised by LTL formulas of size at least $2^{n-1}$. It is recognised by a weak alternating  automaton with $2n$ states and size in $O(n)$, using an automaton based on  Leiss's construction~\cite{Ern81}.
 Intuitively, the alternating automaton represents an $n$-bit up-counter with two states for each bit, one for $1$ and one for $0$ (see \cref{fig:AltCounter}), where the universal transitions enforce that nondeterministic transitions correctly update the counter.

\begin{restatable}[Adaptation of {\cite[proof of Theorem 1]{Ern81}}]{lemma}{AlternatingCounter}
\label{cl:AlternatingCounter}
For every $n\in\Nat\setminus\{0\}$, there is a  weak alternating  automaton with $2n$ states and transition function of size in $O(n)$ recognising the language $\{a^{2^{n-1}}\}$.
\end{restatable}

\begin{figure}[th]
	\centering
	\begin{tikzpicture}[->,>=stealth',shorten >=1pt,auto,node distance=3cm, semithick, initial text=, every initial by arrow/.style={|->},state/.style={circle, draw, minimum size=0.2cm}]
		
		\node[state] (q40) {$q_{4\ZB}$};
		\node[state] (q30) [right of=q40, xshift=1.25cm] {$q_{3\ZB}$};
		\node[state] (q20) [right of=q30, xshift=1.25cm] {$q_{2\ZB}$};
		\node[state] (q10) [right of=q20, xshift=-0.5cm] {$q_{1\ZB}$};
		\node[state] (q41) [accepting,below of=q40,yshift=-2cm]{$q_{acc}$};
		\node[state] (q31) [right of=q41, xshift=1.25cm] {$q_{3\OB}$};
		\node[state] (q21) [right of=q31, xshift=1.25cm] {$q_{2\OB}$};
		\node[state] (q11) [right of=q21, xshift=-0.5cm] {$q_{1\OB}$};

		\node[state] (q111) [above of=q11, yshift=-1.5cm]{\OrState};
		\node (Toqacc) [below of = q111, xshift=-1cm, yshift=2.25cm]{$q_{acc}$};
		\node (Toqacca) [right of = Toqacc, xshift=-3.15cm, yshift=.00cm]{};

		\node[state] (q401) [below of=q40, yshift=1.5cm]{\OrState};
		\node[state] (q403) [right of=q401, xshift=-2cm]{\AndState};
		
		\node[state] (q433) [below of=q403, yshift=2cm]{\OrState};
		
		\node (Toq11) [right of = q433, xshift=-2cm, yshift=-.15cm]{$q_{1\OB}$};
		\node (Toq11a) [right of = Toq11, xshift=-3.15cm, yshift=.00cm]{};
		\node (Toq21) [right of = q433, xshift=-2cm, yshift=-.5cm]{$q_{2\OB}$};
		\node (Toq21a) [right of = Toq21, xshift=-3.15cm, yshift=.00cm]{};
		\node (Toq31) [right of = q433, xshift=-2cm, yshift=-.85cm]{$q_{3\OB}$};
		\node (Toq31a) [right of = Toq31, xshift=-3.15cm, yshift=.00cm]{};

		\node[state] (q301) [below of=q30, yshift=1.5cm]{\OrState};
		\node[state] (q302) [left of=q301, xshift=2cm]{\AndState};
		\node[state] (q303) [right of=q301, xshift=-2cm]{\AndState};

		\node[state] (q311) [above of=q31, yshift=-1.5cm]{\OrState};
		\node[state] (q312) [left of=q311, xshift=2cm]{\AndState};
		\node[state] (q313) [right of=q311, xshift=-2cm]{\AndState};
		
		\node[state] (q332) [below of=q302, yshift=2cm]{\AndState};
		\node[state] (q333) [below of=q303, yshift=2cm]{\OrState};
		
		\node (Toq11m) [right of = q333, xshift=-2cm, yshift=-.15cm]{$q_{1\OB}$};
		\node (Toq11ma) [right of = Toq11m, xshift=-3.15cm, yshift=.00cm]{};
		\node (Toq21m) [right of = q333, xshift=-2cm, yshift=-.5cm]{$q_{2\OB}$};
		\node (Toq21ma) [right of = Toq21m, xshift=-3.15cm, yshift=.00cm]{};
		
		\node (Toq10m) [right of = q332, xshift=-2cm, yshift=.15cm]{$q_{1\ZB}$};
		\node (Toq10ma) [right of = Toq10m, xshift=-3.15cm, yshift=.00cm]{};
		\node (Toq20m) [right of = q332, xshift=-2cm, yshift=.5cm]{$q_{2\ZB}$};
		\node (Toq20ma) [right of = Toq20m, xshift=-3.15cm, yshift=.00cm]{};

		\node (Toqaccm) [below of = q312, xshift=0cm, yshift=2.25cm]{$q_{acc}$};
		\node (Toqaccma) [right of = Toqaccm, xshift=-3.15cm, yshift=.00cm]{};
		
		\node[state] (q201) [below of=q20, yshift=1.5cm]{\OrState};
		\node[state] (q202) [left of=q201, xshift=2cm]{\AndState};
		\node[state] (q203) [right of=q201, xshift=-2cm]{\AndState};
		
		\node[state] (q211) [above of=q21, yshift=-1.5cm]{\OrState};
		\node[state] (q212) [left of=q211, xshift=2cm]{\AndState};
		\node[state] (q213) [right of=q211, xshift=-2cm]{\AndState};
		
		\node[state] (q232) [below of=q202, yshift=2cm]{\AndState};
		\node[state] (q233) [below of=q203, yshift=2cm]{\OrState};
		
		\node (Toq11t) [right of = q233, xshift=-2cm, yshift=-.5cm]{$q_{1\OB}$};
		\node (Toq11ta) [right of = Toq11t, xshift=-3.15cm, yshift=.00cm]{};
		
		\node (Toq10t) [right of = q232, xshift=-2cm, yshift=.3cm]{$q_{1\ZB}$};
		\node (Toq10ta) [right of = Toq10t, xshift=-3.15cm, yshift=.00cm]{};
		
		\node (Toqacct) [below of = q212, xshift=0cm, yshift=2.25cm]{$q_{acc}$};
		\node (Toqaccta) [right of = Toqacct, xshift=-3.15cm, yshift=.00cm]{};

		\path 

		(q10) edge [->,out=-60, in=60,looseness=0.5]  node {a} (q11)
		(q11) edge [->]  node {a} (q111)
		(q111) edge [->] (q10)
		(q111) edge [->] (Toqacca)

		(q40) edge [->] node {a} (q401)
		(q401) edge [->] (q41)
		(q401) edge [->] (q403)

		(q403) edge [->] (q40)
		(q403) edge [->] (q433)

		(q433) edge [->] (Toq11a)
		(q433) edge [->] (Toq21a)
		(q433) edge [->] (Toq31a)

		(q30) edge [->] node {a} (q301)
		(q301) edge [->] (q302)
		(q301) edge [->] (q303)

		(q303) edge [->] (q30)
		(q303) edge [->] (q333)
		
		(q302) edge [->,out=-120, in=180,looseness=1.0] (q31)
		(q302) edge [->] (q332)
		
		(q333) edge [->] (Toq11ma)
		(q333) edge [->] (Toq21ma)
		
		(q332) edge [->] (Toq10ma)
		(q332) edge [->] (Toq20ma)

		(q31) edge [->] node {a} (q311)
		(q311) edge [->] (q312)
		(q311) edge [->] (q313)
		(q311) edge [->] (Toqaccma)
		
		(q313) edge [->] (q31)
		(q313) edge [->] (q333)
		
		(q312) edge [->,out=120, in=180,looseness=1.0] (q30)
		(q312) edge [->] (q332)
		
		(q20) edge [->] node {a} (q201)
		(q201) edge [->] (q202)
		(q201) edge [->] (q203)
		
		(q21) edge [->] node {a} (q211)
		(q211) edge [->] (q212)
		(q211) edge [->] (q213)
		(q211) edge [->] (Toqaccta)
		
		(q203) edge [->] (q20)
		(q213) edge [->] (q21)
		(q203) edge [->] (q233)
		(q213) edge [->] (q233)
		
		(q202) edge [->,out=-120, in=180,looseness=1.0] (q21)
		(q212) edge [->,out=120, in=180,looseness=1.0] (q20)
		(q202) edge [->] (q232)
		(q212) edge [->] (q232)
		
		(q233) edge [->] (Toq11ta)
		
		(q232) edge [->] (Toq10ta)
		
				;
	\end{tikzpicture}
	\caption{A weak alternating automaton of size in $O(n)$ recognising $\{a^{2^{n-1}}\}$; here with $n=4$, where the initial configuration is $q_{1\ZB} \land q_{2\ZB} \land q_{3\ZB} \land q_{4\ZB}$.}\label{fig:AltCounter}
\end{figure}
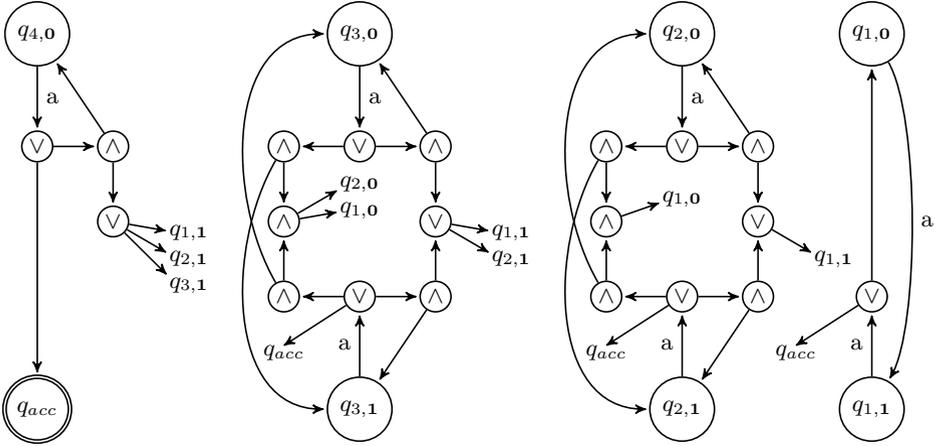
 
 We continue to nondeterministic automata (NFAs), for which the arguments are more involved as they do not allow for linear complementation.

\begin{restatable}{lemma}{NFAtoLTL}\label{cl:NFAtoLTL}
Every nondeterministic automaton with $n$ states recognising an LTL-expressible language $L\subseteq \{a\}^+$ is equivalent to an LTL formula of size in $O(n^{2})$.
\end{restatable}

\begin{proof}[Proof sketch] For finite $L$, by a pumping argument, $\A$ only accepts words up to length $n$, and by \cref{cl:LTLforShortLang} we are done.
We now consider a co-finite $L$. 

We use 2-way deterministic automata, which are deterministic automata that process words of the form $\vdash\!\!w\!\!\dashv$, where $\vdash$ and $\dashv$ are start- and end-of-word markers respectively, and where transitions specify whether to read the letter to the right or to the left of the current position. They accept by reaching an end state, and reject by reaching a rejecting state or by failing to terminate~\cite{GMP07}, and every unary NFA $\A$ can be turned into a 2-way DFA $\D$ of size $O(n^2)$~\cite{Chr86}. 

We construct from an NFA $\A$ a 2-way DFA $\D$, and then a 2-way DFA $\D'$ of the same size that recognises $a^*\setminus \{a^k\}$, where $a^k$ is the longest word not in $L$. 
We use the fact that a 2-way DFA of size $m$ can be complemented into one of size $4m$ \cite{GMP07} to complement $\D'$ into $\D''$ that recognises $\{a^k\}$ and must therefore be of size at least $k+2$~\cite{Bir96}, so $k$, and by Proposition~\ref{cl:AlternatingCounter}, an LTL formula for $L$, is in $O(n^2)$.
\end{proof}

We now show that this upper bound is tight.
The previous lower bound ideas do not work with nondeterminism, since we need $n$ states to recognise $\{a^{n}\}$~\cite{KTV99}.
Yet, we need not count \textit{exactly} to $n$ for achieving a lower bound. We can use a variant of a language used in~\cite[pages 10--11]{BK10}:
For every positive integer $k$, define the set of positive integers $S_k = \{m > 0 \mid \exists i,j\in \mathbb{N}.~m=ik+j(k+1)\}$, and the language $V_k= \{a^m \mid m \in S_k\}\subseteq \{a\}^*$.

\begin{proposition}[{Folklore, \cite[Theorem 3]{BK10}}]\label{cl:numbers}
 For every $k\in \mathbb{N}$ the number $k^2 - k - 1$ is the maximal number not in $S_k$. 
\end{proposition}

\begin{proposition}[{\cite[proof of Theorem 4]{BK10}}]\label{cl:NFAforSquare}
	For every $n\in\Nat$, there is an NFA of size in $O(n)$ recognising a co-finite language $L\subseteq\{a\}^*$, such that $a^{k^2 - k - 1}$ is not in $L$, while for every $t \geq k^2 - k$, we have that $a^t\in L$.
\end{proposition}

\begin{restatable}{theorem}{UnaryBlowup}
	The size blow-up involved in translating deterministic, nondeterministic, and alternating automata on finite unary words to LTL, when possible, is $\Theta(n)$, $\Theta(n^2)$, and $\Theta(2^n)$, respectively.
\end{restatable}
\section{General Alphabet}\label{sec:General}

In this section we consider the more challenging task of turning counter-free $\omega$-regular automata over arbitrary alphabets into LTL. We use the fact that these automata can be  turned into reset cascade automata (Krohn-Rhodes-Holonomy decomposition), which we describe in Section~\ref{sec:CascadedAutomata}. Our technical contribution is then the translation of reset cascade automata into LTL.

In brief, we build, in~\cref{sec:ReachbilityFormulas}, a \emph{parameterised LTL formula} that is satisfied by a word $w$ iff the run of the cascade on $w$, starting in the parameter configuration $S$, reaches a parameter configuration $T$, such that the remaining suffix of $w$ satisfies a parameter LTL formula $\tau$. 
We then use this formula, in~\cref{sec:DetAutomataToLTL}, to describe the automaton's acceptance condition.

When encoding the behavior of a cascade by an LTL formula, we need to overcome two major challenges: First, the cascade is a formalism that looks at the \emph{past}, namely at the word read so far, to determine the next configuration, while an LTL formula obtains its value only from the future. Second, the cascade has an internal state, while an LTL formula does not. Our reachability formulas are therefore quite involved, built inductively over the number of levels in the cascade, and implicitly allowing to track the internal configuration of the cascade.

In \cref{sec:SizeAnalysis} we analyse the length and depth of the resulting formulas.

\subsection{Cascaded Automata}\label{sec:CascadedAutomata}
\Subject{Cascades}
A cascaded semiautomaton (analogous to the algebraic wreath pro-duct) over an alphabet $\Sigma$ is a semiautomaton that can be described as a sequence of simple semiautomata, such that the alphabet of each of them is $\Sigma$ together with the current state of each of the preceding semiautomata in the sequence. It is a reset cascade if it is a sequence of reset semiautomata.
Formally, a \emph{cascaded semiautomaton}, or just \emph{cascade}, over alphabet $\Sigma$ with $n$ levels is a tuple $\A = \tuple{\Sigma, \A_1, \A_2, \dots, \A_n}$, such that $\A_i = (\Sigma_i, Q_i, \delta_i)$ is a semiautomaton for each level $i$, where $\Sigma_i = \Sigma \times Q_1 \times \cdots \times Q_{i{-}1}$. (So $\Sigma_1=\Sigma$, $\Sigma_2=\Sigma\times Q_1$, etc.). It is a \emph{reset cascade} if all $\A_i$'s are reset semiautomata.

An $i$-configuration $S$ of $\A$ is a tuple $\tuple{q_1, q_2, \ldots, q_i} \in Q_1 \times \cdots \times Q_i$. If $q_{i+1} \in Q_{i+1}$ is a state of level $i+1$, we write $\tuple{S, q_{i+1}}$ for the $(i+1)$-configuration $\tuple{q_1,\ldots,q_i,q_{i+1}}$. Note that the $0$-configuration is the empty tuple $\tuple{}$.
 Further, we derive the  transition relation for configurations by point-wise application of the respective $\delta_i$'s. We define $\delta_{\leq i}(\tuple{q_1, q_2, \dots q_i}, \sigma)$ as $\tuple{\delta_1(q_1, \tuple{\sigma}), \delta_2(q_2, \tuple{\sigma, q_1}), \dots}$. Note that we will omit the ``$\leq i$''-subscript if it is clear from context, and by just writing ``configuration'', we mean an $n$-configuration.

Notice that $\A$ describes a standard semiautomaton $\D_\A$ over $\Sigma$, whose states are the configurations of $\A$ of level $n$, and its transition function is $\delta_{\leq n}$. If there are up to $j$ states in each level of $\A$, there are up to $j^n$ states in $\D_\A$.
Observe that when $\A$ is a reset cascade, it can be translated to an equivalent reset cascade with up to $n\log j$ levels, and $2$ states in each level \cite[Ex. I.10.2]{Eil76}.

For a state $q \in Q_i$ of level $i$ of a reset cascade, we denote by $\Enter{q}$, $\Stay{q}$, and $\Leave{q} \subseteq \Sigma \times Q_1 \times \cdots \times Q_{i{-}1}$ the sets of (combined) letters that enter $q$, stay in it, and leave it, respectively. These are sets of pairs $\tuple{\sigma, S}$, where $S$ is an $(i{-}1)$-configuration and $\sigma \in \Sigma$. Notice that $\Enter{q} \subseteq \Stay{q}$, and that $\Leave{q}$ is the complement of $\Stay{q}$ (w.r.t.\ the relevant (combined) letters).

A semiautomaton $(\Sigma,Q,\delta)$ is \emph{homomorphic} to a cascade $\tuple{\Sigma, \A_1, \dots, \A_n}$ if there exists  a partial surjective function $\phi \colon Q_1 \times \cdots \times Q_{n} \to Q$, such that for every $\sigma\in\Sigma$ and $S\in Q_1 \times \cdots \times Q_{n}$, we have $\delta(\phi(S),\sigma) = \phi(\delta_{\leq n}(S,\sigma))$.

\begin{proposition}[Part of the Krohn-Rhodes-Holonomy Decomposition
	 {\cite[Corollary II.7.2]{Eil76}\label{cl:KrohnRhodes}, \cite[Theorem 3]{MP90}}]
	Every counter-free deterministic semiautomaton $\D$ with $n$ states is homomorphic to a reset cascade $\A$ with up to $2^n$ levels and $2^n$ states in each level.
\end{proposition}

\begin{remark}\label{rem:Cascade}
The Krohn-Rhodes and Holonomy decomposition theorems consider also more general cascades and give results with respect to arbitrary semiautomata. The Holonomy decomposition in \cite{Eil76}, as opposed to many other proofs of the Krohn-Rhodes decomposition, guarantees up to $2^n$ levels with up to $2^n$ states in each level. Yet, it shows that $\A$ \emph{covers} $\D$, allowing $\A$ to operate over an alphabet different from that of $\D$. In \cite{MP90,MP94,Mal10}, the algebraic proof of \cite{Eil76} is translated to an automata-theoretic one, providing the stated homomorphism. It is also stated in \cite[Theorem 3.1]{MP90}, \cite[Corollary 20]{MP94}, and \cite[Corollary 2]{Mal10} that the number of configurations in $\A$ is singly exponential in $n$, but to the best of our understanding they do not provide an explicit proof for it.
\end{remark}

\Subject{Cascades with acceptance conditions}
As a cascade $\A$ describes a standard semiautomaton (whose states are the configurations of $\A$), we can add to it an initial configuration and an acceptance condition to make it a standard deterministic automaton.
We show below that the homomorphism between an automaton and a cascade can be extended to also transfer the same acceptance condition.

\begin{restatable}{proposition}{RabinCascade}
\label{cl:ResetCascadeRabin}
	Let $\D$ be a deterministic B\"uchi, coB\"uchi or Rabin automaton, with a semiautomaton homomorphic to a cascade $\A$. There is respectively a deterministic B\"uchi, coB\"uchi or Rabin automaton $\D'$ equivalent to $\D$ with semiautomaton $\A$. For Rabin, $\D$ and $\D'$ have the same number of acceptance pairs.
\end{restatable}

\begin{restatable}{proposition}{MullerCascade}
\label{cl:ResetCascadeMuller}
	Consider a deterministic Muller automaton $\D$ with $n$ states, whose semiautomaton is homomorphic to a reset cascade $\A$ with $m$ configurations. Then there is a deterministic Muller automaton $\D'$ equivalent to $\D$, whose semiautomaton is $\A$ and its Muller condition has up to $2^{O(mn)}$ acceptance sets.
\end{restatable}

\subsection{Encoding Reachability within Reset Cascades by LTL Formulas}\label{sec:ReachbilityFormulas}

For the rest of this section, let us fix a set of atomic propositions $\AP$, an alphabet $\Sigma = 2^{\AP}$, and a reset cascade $\A = \tuple{\Sigma, \A_1, \A_2, \dots, \A_n}$.

\Subject{The main reachability formula} For every level $i$ of $\A$, three configurations $S,B$ and $T$
 of level $i$, and two LTL formulas $\beta$ and $\tau$, we will define the LTL formula \smallerReach{$\Reach{S}{B}{\beta}{T}{\tau}$} with the intended semantics that it holds on a word $w \in \Sigma^\omega$ iff $\A$ goes from the `starting' configuration $S$ to the `target' configuration $T$ along some prefix $u$ of $w$, such that the suffix of $w$ after $u$ satisfies  $\tau$ and  the path along $u$ avoids the `bad' configuration $B$ with a suffix satisfying $\beta$. 

\Subject{Auxiliary reachability formulas} We will formally define the main reachability formula by induction on the level $i$ of the involved configurations, and using four auxiliary formulas, whose intended semantics is described in \cref{table:reach-formulas}.  These formulas distinguish between the case that the top-level state is unchanged along the reachability path, denoted with a solid arrow $\xrightarrow{\quad}$, and the case that it is changed, denoted by a dashed arrow $\xdashrightarrow{\quad}$. They also have dual, weak, versions.

Observe that intuitively \smallerReach{$\Reach{S}{B}{\beta}{T}{\tau}$} is an extended \emph{Until} operator, while its dual \smallerReach{$\WeakReach{S}{B}{\beta}{T}{\tau} = \neg (\Reach{S}{T}{\tau}{B}{\beta})$} is an extended \emph{Weak until} (or \emph{Release}) operator. We build the formulas so that for appropriate choices of $\beta$ and $\tau$, the (strong) reachability formulas 1, 3, and 5 (as numbered in \cref{table:reach-formulas}) are syntactic co-safety and the weak formulas 2 and 4 are syntactic safety formulas. 

\begin{table}[ht!]
\begin{changemargin}{-2cm}{-2cm}
{\footnotesize
\[\begin{array}{lc|rl}
\multicolumn{2}{c|}{\multirow{3}{*}{\text{Reachability formula $\phi$}}} & &\text{~~~~~~~~Intended semantics}\\
&& \Intui{Intuitively:} &~  \Intui{Reading a word $w$ from the configuration $S$ or $\tuple{S,s}$}\\
&& \text{Formally:} &~ w \models \varphi \iff \\[0.5em]
\hline
&&&\\
\multirow{3}{*}{1.}&\multirow{3}{*}{\Reach{S}{B}{\beta}{T}{\tau}} ~ &
\multicolumn{2}{l}{\Intui{ not reaching  $B(\beta)$ until reaching $T(\tau)$.}} \\
&	&         \exists i \geq 0.   & \delta(S,w_{[0..i)}) = T \land w_{[i..]} \models \tau \\
 & 	&                             & \land ~ (\forall j \in [0..i). ~ \delta(S,w_{[0..j)}) \neq B \lor w_{[j..]} \not \models \beta) \\[1.4em]

\multirow{3}{*}{2.}&\multirow{3}{*}{\WeakReach{S}{B}{\beta}{T}{\tau}} ~ &
\multicolumn{2}{l}{\Intui{ reaching $T(\tau)$ releases not reaching $B(\beta)$.}} \\
&    &         \forall i \geq 0.   & (\delta(S,w_{[0..i)}) = B \land w_{[i..]} \models \beta) \\
&  	&                             & \rightarrow ~ (\exists j \in [0..i). ~ \delta(S,w_{[0..j)}) = T \land w_{[j..]} \models \tau) \\[1.4em]

\multirow{4}{*}{3.}&\multirow{4}{*}{\StayReach{\tuple{S,s}}{\tuple{B,b}}{\beta}{\tuple{T,t}}{\tau}} ~ &
\multicolumn{2}{l}{\Intui{ not reaching  $\tuple{B,b}(\beta)$ until reaching $\tuple{T,t}(\tau)$, while staying in $s$.}} \\
&	&         \exists i \geq 0.   & \delta(\tuple{S,s},w_{[0..i)}) = \tuple{T,t} \land w_{[i..]} \models \tau \\
&  	&	       	  	              & \land ~ (\forall j \in [0..i). ~ \delta(\tuple{S,s},w_{[0..j)}) \neq \tuple{B,b} \lor w_{[j..]} \not \models \beta) \\
&	& 			                  &\textcolor{orange}{ \land ~ (\forall j \in [0..i). ~ \tuple{w[j], \delta(S, w_{[0..j)})} \in \Stay{s})} \\[1.4em]

\multirow{4}{*}{4.}&\multirow{4}{*}{\WeakStayReach{\tuple{S,s}}{\tuple{B,b}}{\beta}{\tuple{T,t}}{\tau}} ~&
\multicolumn{2}{l}{\Intui{ reaching $\tuple{T,t}(\tau)$ releases not (reaching $\tuple{B,b}(\beta)$ or leaving $s$).}} \\
&    &           \forall i \geq 0. & \big((\delta(\tuple{S,s},w_{[0..i)}) = \tuple{B,b} \land w_{[i..]} \models \beta) \\
&    &                             & ~ ~ ~  \textcolor{orange}{\lor ~(i > 0 \land \tuple{w[i{-}1], \delta(S, w_{[0..i{-}1)})} \in \Leave{s})} \big) \\
&  	&                             & \rightarrow ~ (\exists j \in [0..i). ~ \delta(\tuple{S,s},w_{[0..j)}) = \tuple{T,t} \land w_{[j..]} \models \tau) \\[1.4em]

\multirow{6}{*}{5.}&\multirow{6}{*}{\LeaveReach{\tuple{S,s}}{\tuple{B,b}}{\beta}{\tuple{T,t}}{\tau}} ~&
\multicolumn{2}{l}{\Intui{ not reaching  $\tuple{B,b}(\beta)$ until reaching $\tuple{T,t}(\tau)$ and leaving $s$.}} \\
&    &  ~ \exists i_1, \textcolor{orange}{i_2} \geq 0. & \delta(\tuple{S,s},w_{[0..i_1)}) = \tuple{T,t} \land w_{[i_1..]} \models \tau \\
&    &                             & \textcolor{orange}{\land ~ (\exists j_1 \in [0..i_1). ~ \tuple{w[j_1], \delta(S,w_{[0..j_1)})} \in \Enter{t})} \\
&    &                             & \textcolor{orange}{\land ~ \tuple{w[i_2], \delta(S,w_{[0..i_2)})} \in \Leave{s}} \\
&    &			                  & \land ~ (\forall j_2\in [0..\max(i_1{-}1, i_2)]. ~ \delta(\tuple{S,s},w_{[0..j_2)}) \neq \tuple{B,b}\\ 
&   &						 	& \hspace{3.9cm} \lor~ w_{[j_2..]} \not \models \beta) \\
\end{array}\]}
\end{changemargin}
\caption{The intended semantics of reachability formulas. \textcolor{orange}{Orange subformulas} show the difference between the auxiliary formulas and the first or second (main) formula.}
\label{table:reach-formulas}
\end{table}
	
{\allowdisplaybreaks
\Subject{Formulas 1 and 2}
The main formula is simply defined as the union of two auxiliary formulas, corresponding to whether or not the top-level state changes, and its weak version is defined to be its dual.
\begin{align*}
\Reach{S}{B}{\beta}{T}{\tau} \coloneqq & \begin{cases}
(\neg \beta) \U \tau & \text{if } S = \tuple{} \\
\StayReach{S}{B}{\beta}{T}{\tau} \vee \LeaveReach{S}{B}{\beta}{T}{\tau} & \text{otherwise.}
\end{cases} \\
\WeakReach{S}{B}{\beta}{T}{\tau} \coloneqq & ~ \neg \left(\Reach{S}{T}{\tau}{B}{\beta}\right)
\end{align*}

\Subject{Formula 3}
Since the formula should ensure that the top-level state $s$ is unchanged, we first distinguish between four cases, depending on which of the source configuration $\tuple{S,s}$, bad configuration $\tuple{B,b}$, and target configuration $\tuple{T,t}$ are equal.
The definitions of the four cases only differ in whether or not each of $\beta$ and $\tau$ are satisfied in the first position of the word. 

We  define them using an intermediate common formula that is indifferent to the first position, which we mark by ``$\FromSecond$'' on top of the arrow.
We then define the ``$\FromSecond$'' formula by using the main reachability formula with respect to a lower level, namely with respect to the configurations $S$ and $T$ instead of $\tuple{S,s}$ and $\tuple{T,t}$, and having corresponding disjunctions and conjunctions on all the combined letters of the top level that belong to $\Stay{s}$ and $\Leave{s}$.

{\footnotesize
\begin{align*}
& \StayReach{\tuple{S,s}}{\tuple{B,b}}{\beta}{\tuple{T,t}}{\tau} \coloneqq \\
& \quad \qquad \begin{cases}
\StayReachNonEmpty{\tuple{S,s}}{\tuple{B,b}}{\beta}{\tuple{T,t}}{\tau}
  & \text{if } \tuple{S,s} \neq \tuple{B,b} \text{ and } \tuple{S,s} \neq \tuple{T,t} \\
\StayReachNonEmpty{\tuple{S,s}}{\tuple{B,b}}{\beta}{\tuple{T,t}}{\tau} \lor \tau
  & \text{if } \tuple{S,s} \neq \tuple{B,b} \text{ and } \tuple{S,s}    = \tuple{T,t} \\
\StayReachNonEmpty{\tuple{S,s}}{\tuple{B,b}}{\beta}{\tuple{T,t}}{\tau} \land \neg \beta
  & \text{if } \tuple{S,s}    = \tuple{B,b} \text{ and } \tuple{S,s} \neq \tuple{T,t} \\
\left(\StayReachNonEmpty{\tuple{S,s}}{\tuple{B,b}}{\beta}{\tuple{T,t}}{\tau}  \land \neg \beta\right) \lor \tau
  & \text{if } \tuple{S,s}    = \tuple{B,b} \text{ and } \tuple{S,s}    = \tuple{T,t}
\end{cases} \\
&\text{where }  \StayReachNonEmpty{\tuple{S,s}}{\tuple{B,b}}{\beta}{\tuple{T,t}}{\tau} \coloneqq
 \hspace{-1cm}\quad \qquad \bigvee_{\substack{\tuple{\sigma,T'}\in\Stay{s} \\\text{s.t. }\tuple{T'\!,s}\overset{\sigma}{\to} \tuple{T,t}}} 
 \hspace{-0.2cm}
\Bigg( \Reach{S}{S}{\false}{T'}{\sigma \land \X \tau}  \\ 
&
\hspace{2em}\land \bigwedge_{\tuple{\eta,L}\in\Leave{s}} \hspace{0.2em}
\hspace{-1em}\Reach{S}{L}{\eta}{T'}{\sigma \land \X \tau}
 \hspace{1em}\land\hspace{-0.5em} \bigwedge_{\substack{\tuple{\rho,B'}\in\Stay{s}\\\text{s.t. }\tuple{B'\!,s}\overset{\rho}{\to} \tuple{B,b}}}\hspace{-1em}
\Reach{S}{B'}{\rho \land \X \beta}{T'}{\sigma \land \X \tau}\Bigg) 
\end{align*}}

\Subject{Formula 4}
Its intended semantics is also that the top-level state $s$ is unchanged, but we weaken Formula 3 by not enforcing that the target configuration $\tuple{T,t}$ is reached and $\tau$ is satisfied. Thus as long as the top-level state $s$ stays unchanged and the bad configuration $\tuple{B,b}$ is not reached while satisfying $\beta$, Formula 4 is also satisfied. Note that since both Formula 3 and Formula 4 need to ensure that the top-level state $s$ is unchanged they cannot simply be defined as the dual of each other. However, they share the same construction principle:

{
\footnotesize
\begin{align*}
& \WeakStayReach{\tuple{S,s}}{\tuple{B,b}}{\beta}{\tuple{T,t}}{\tau} \coloneqq \\
& \quad \qquad \begin{cases}
\WeakStayReachNonEmpty{\tuple{S,s}}{\tuple{B,b}}{\beta}{\tuple{T,t}}{\tau}
  & \text{if } \tuple{S,s} \neq \tuple{B,b} \text{ and } \tuple{S,s} \neq \tuple{T,t} \\
\WeakStayReachNonEmpty{\tuple{S,s}}{\tuple{B,b}}{\beta}{\tuple{T,t}}{\tau} \lor \tau
  & \text{if } \tuple{S,s} \neq \tuple{B,b} \text{ and } \tuple{S,s}    = \tuple{T,t} \\
\WeakStayReachNonEmpty{\tuple{S,s}}{\tuple{B,b}}{\beta}{\tuple{T,t}}{\tau} \land \neg \beta
  & \text{if } \tuple{S,s}    = \tuple{B,b} \text{ and } \tuple{S,s} \neq \tuple{T,t} \\
\left(\WeakStayReachNonEmpty{\tuple{S,s}}{\tuple{B,b}}{\beta}{\tuple{T,t}}{\tau}  \lor \tau \right) \land \neg \beta
  & \text{if } \tuple{S,s}    = \tuple{B,b} \text{ and } \tuple{S,s}    = \tuple{T,t}
\end{cases} \\
& \text{where} \\[0.5em]
& \WeakStayReachNonEmpty{\tuple{S,s}}{\tuple{B,b}}{\beta}{\tuple{T,t}}{\tau} \coloneqq \\
& \hspace{-1cm}\quad \qquad \bigvee_{\substack{\tuple{\sigma,T'}\in\Stay{s} \\\text{s.t. }\tuple{T'\!,s}\overset{\sigma}{\to} \tuple{T,t}}} \hspace{-0.1cm}
\Bigg(
\hspace{0em}
\bigwedge_{\tuple{\eta,L}\in\Leave{s}}
\hspace{-0.4em} \WeakReach{S}{L}{\eta}{T'}{\sigma \land \X \tau}
  \land \hspace{-1em}
\bigwedge_{\substack{\tuple{\rho,B'}\in\Stay{s}\\\text{s.t. }\tuple{B'\!,s}\overset{\rho}{\to} \tuple{B,b}}}
\hspace{-0.3cm}
\WeakReach{S}{B'}{\rho \land \X \beta}{T'}{\sigma \land \X \tau}\Bigg) & (1) \\
& \hspace{2.5em} \bigvee
\hspace{2.1em}
\Bigg(
\hspace{0em}
\bigwedge_{\substack{\tuple{\eta,L}\in\Leave{s}}}
\hspace{-0.4em} \WeakReach{S}{L}{\eta}{S}{\false} 
\hspace{1.2em} \land \hspace{-1em}
\bigwedge_{\substack{\tuple{\rho,B'}\in\Stay{s}\\\text{s.t. }\tuple{B'\!,s}\overset{\rho}{\to} \tuple{B,b}}}
\hspace{-0.3cm}
\WeakReach{S}{B'}{\rho \land \X \beta}{S}{\false}\Bigg) & (2) \\
\end{align*}
}
\vspace{-1cm}
\Subject{Formula 5}
The definition of the last reachability formula is the most challenging, since the top-level state changes ($s \neq t$), which prevents the direct usage of lower level configurations.

Intuitively, before reaching the target configuration $\tuple{T,t}$, the run must see a combined letter $\tuple{\sigma, T'} \in \Enter{t}$, after which the top-level state $t$ is preserved and the bad situation $\tuple{B,b}(\beta)$ is avoided.
This is line (1) of the definition.

The run must also not see $\tuple{B,b}(\beta)$ before reaching $T'$, which is handled in line (2), whose difference from line (1) is the additional constraint on the path from $S$ to $T'$. (Line (1) is required for the case that $\Enter{b}$ is empty.)
We use Formula 4 for that constraint, rather than Formula 3 which could also be used, in order to ensure that Formula 5 can be a syntactic co-safety formula.

Lastly, line (3) ensures that the top-level state is indeed changed.

{
\footnotesize
\begin{align*}
& \LeaveReach{\tuple{S,s}}{\tuple{B,b}}{\beta}{\tuple{T,t}}{\tau} \coloneqq \\
&  \bigvee_{\substack{\tuple{\sigma,T'}\in\\\Enter{t}}} \Bigg(
\Reach{S}{S}{\false}{T'}{\sigma \land \X \Big( \StayReachMulti{\delta(\tuple{T',\cdot},\sigma)}{\tuple{B,b}}{\beta}{\tuple{T,t}}{\tau}{} \Big) } ~ \land & (1) \\
             &  \hspace{0em} \bigwedge_{\substack{\tuple{\eta,R}\in\\\Enter{b}}} \hspace{-0.0cm}
                		\Reach{S}{R}{\eta \land \X ( \StayReachMulti{\delta(\tuple{R,\cdot},\eta)}{\tuple{T,t}}{\tau}{\tuple{B,b}}{\beta}{\Weak} )}{T'}{\sigma \land \X \Big( \StayReachMulti{\delta(\tuple{T',\cdot},\sigma)}{\tuple{B,b}}{\beta}{\tuple{T,t}}{\tau}{} \!\Big)\!} \!\!\! \Bigg) & (2) \\
& \hspace{0em} \land ~ \bigvee_{\substack{\tuple{\sigma,L}\in\\\Leave{s}}}  \StayReach{\tuple{S,s}}{\tuple{B,b}}{\beta}{\tuple{L,s}}{\sigma \land
	 \begin{cases} \neg \beta & \text{if } \tuple{L,s} = \tuple{B,b} \\ \true & \text{otherwise.} \end{cases}} & (3) 
\end{align*}
}}

We prove the correctness of the above definitions with respect to the intended meaning of \Cref{table:reach-formulas} by induction on the level of the involved configurations.
\begin{restatable}{lemma}{thmCorrectness}\label{thm:reach-formulas-correctness}
	The intended semantics of \Cref{table:reach-formulas} hold for all infinite words $w \in \Sigma^\omega = (2^{\AP})^\omega$, configurations $S,B,T$ of level $m \leq n$, states $s,b,t$ in level $m+1$ (when $m < n$), and LTL formulas $\beta$ and $\tau$ over $\AP$.
\end{restatable}

Using the same induction principle we prove that the reachability formulas stay within certain classes of the syntactic future hierarchy (\cref{def:future_hierarchy}). We use \smallerReach{$\Reach{S}{B}{X}{T}{Y} \in Z$} as a shorthand for saying that for every formulas $\beta \in X$ and $\tau \in Y$, the formula \smallerReach{$\Reach{S}{B}{\beta}{T}{\tau}$} is in $Z$.

\begin{restatable}{lemma}{lemSyntacticSafe}\label{lem:reach-formulas-syntactic}
Let $S,B$, $T$ be configurations of level $m \leq n$, and let $s,b,t$ be states in level $m+1$ (when $m < n$). Then for $i \geq 1$ it holds that:
\begin{itemize}
	\item \smallerReach{$\Reach{S}{B}{\Pi_i}{T}{\Sigma_i}, ~\StayReach{\tuple{S,s}}{\tuple{B,b}}{\Pi_i}{\tuple{T,t}}{\Sigma_i}, ~\LeaveReach{\tuple{S,s}}{\tuple{B,b}}{\Pi_i}{\tuple{T,t}}{\Sigma_i} ~\in ~\Sigma_i$}
	\item \smallerReach{$\WeakReach{S}{B}{\Sigma_i}{T}{\Pi_i}, ~\WeakStayReach{\tuple{S,s}}{\tuple{B,b}}{\Sigma_i}{\tuple{T,t}}{\Pi_i} ~\in~ \Pi_i $}
	\end{itemize}
\end{restatable}
\subsection{Depth and Length Analysis}\label{sec:SizeAnalysis}
We analyze the length and temporal-nesting depth of the LTL reachability formulas defined in \cref{sec:ReachbilityFormulas}. Notice that both measures are of independent interest, as there might be a non-elementary gap between the depth and length of LTL formulas \cite[Theorem 6]{EVW02}. Since we provide upper bounds, the bound on the length of formulas obviously gives also a bound on their size.

We consider a reset cascade $\A$ with $n$ levels, as in \cref{sec:ReachbilityFormulas}, and further assume for the length and depth analysis that it has up to $n$ states in each level. (This assumption holds in the reset cascades that result from the Krohn-Rohdes decomposition as per \cref{cl:KrohnRhodes}.)

We define for each of the five reachability formulas a \emph{depth function} $\DF_x(i,d)$ and a \emph{length function} $\LF_x(i,l)$, where $x$ refers to the number of the reachability formula, to bound the depth and length of the formulas. These depend on the level $i$ of its input configurations $S, B$ and $T$, and the maximal depth $d$ and length $l$ of its input formulas $\beta$ and $\tau$. For the main (first) reachability formula, we also use $\DF$ and $\LF$, standing for $\DF_1$ and $\LF_1$.
For example, the length of the first formula  $\Reach{S}{B}{\beta}{T}{\tau}$ over configurations $S, B$ and $T$ of level $7$ and formulas $\beta$ and $\tau$ of length up to $77$ is bounded by the value of $\LF_1(7,77)$.

For simplicity, we consider the LTL representation of an alphabet letter $\sigma\in\Sigma$ to be of length $1$, while its actual length is $3 \log_2 |\Sigma|$. This increase is due to the need to encode an alphabet letter $\sigma \in \Sigma = 2^{\AP}$ as a conjunction of atomic propositions in $\AP$. The  representation length can be multiplied by the total length of the final relevant formula (e.g., a formula equivalent to the entire reset cascade), since it remains constant along all steps of our inductive computation.

We provide in \cref{tab:size} upper bounds on the depth and length functions, relative to values of other depth and length functions with respect to configurations of the same or lower-by-one level.
The table is constructed by following the syntactic definitions of the reachability formulas, and applying basic simplifications to the resulting expressions.
For example, $\LF_1(0,l)=2+2l$ standing for the length of $(\neg \beta) \U \tau$.
In \cref{cl:ReachFormulaSize} we will use \cref{tab:size} to bound the absolute depth and length of the main reachability formula.

\begin{table}[h]
\begin{changemargin}{-2cm}{-2cm}
		\[\begin{array}{lc|rl}
		\multicolumn{2}{c|}{\text{Reachability formula } \varphi}~ & \multicolumn{2}{c}{\text{Bounds on } \depth{\phi} \text{ and length } |\varphi|} \\[0.5em]
		\hline
		&& \\
		\multirow{4}{*}{1.} &
		\multirow{4}{*}{\Reach{S}{B}{\beta}{T}{\tau}} ~
		& ~\DF_1(i,d)     \leq & \begin{cases} d + 1 & \text{if } i = 0 \\ \max (\DF_3(i,d) \, , \,\DF_5(i,d)) & \text{otherwise.} \end{cases} \\[1.4em]
		&& ~\LF_1(i,l)~     \leq & \begin{cases} 2 + 2l & \text{if } i = 0 \\ 1 + \LF_3(i,l) + \LF_5(i,l) & \text{otherwise.} \end{cases} \\[2.4em]

		\multirow{2}{*}{2.} &
		\multirow{2}{*}{\WeakReach{S}{B}{\beta}{T}{\tau}} ~
		& \DF_2(i,d)    = & \DF_1(i,d) \\ [0.2em]
		&& \LF_2(i,l)~    = & 1 + \LF_1(i,l) \\ 
		&&\\

		\multirow{3}{*}{3.} &
		\multirow{3}{*}{\StayReach{\tuple{S,s}}{\tuple{B,b}}{\beta}{\tuple{T,t}}{\tau}} ~
		&   \DF_3(i,d)   \leq & \DF_1(i{-}1, d+1) \\ [0.2em]
		&&   \LF_3(i,l)~   \leq & 3{+}2l + |\Sigma|n^{i-1} \big(1 {+} \LF_1(i{-}1,3{+}l)+ \\
		&&& 1+ |\Sigma|n^{i-1} (\LF_1(i{-}1, 3{+}l)+1) +\\
		&&& 1+ |\Sigma|n^{i-1} (\LF_1(i{-}1, 3{+}l)+1)\big) \\
		&&\leq &  3+ 2l+ 4|\Sigma|^2n^{2(i-1)}\LF_1(i{-}1,l{+}3)
		\\[1em]

		\multirow{3}{*}{4.} &
		\multirow{3}{*}{\WeakStayReach{\tuple{S,s}}{\tuple{B,b}}{\beta}{\tuple{T,t}}{\tau}} ~
		&   \DF_4(i,d)   \leq & \DF_2(i{-}1, d+1) = \DF_1(i{-}1, d+1)\\ [0.2em]
		&& \LF_4(i,l)~ \leq & 3 + 2l + (1+|\Sigma|n^{i-1})\big(
		1+|\Sigma|n^{i-1} (1 + \LF_2(i{-}1, l{+}3))\big) \\
		&&            \leq & 3 + 2l + 4|\Sigma|^2n^{2(i-1)} \LF_1(i{-}1, l+3) \\[1em]

		\multirow{3}{*}{5.} &
		\multirow{3}{*}{\LeaveReach{\tuple{S,s}}{\tuple{B,b}}{\beta}{\tuple{T,t}}{\tau}} ~
		&   \DF_5(i,d)   \leq & \DF_1(i{-}1, \max(1+\DF_3(i,d) \, , \, 1+ \DF_4(i,d))) \\ [0.2em]
		&& \LF_5(i,l)~     \leq & |\Sigma|n^{i-1} \cdot \big(\LF_1(i-1,3+\LF_3(i,l))+2+\\
		&&& |\Sigma| n^{i-1}\cdot \big(\LF_1(i-1, \max(3+\LF_3(i,l),3+\LF_4(i,l)))+1\big)\big)\\
		&&& +1+ |\Sigma|n^{i-1}\cdot (1+\LF_3(i,3+l))\\
	\end{array}\]
\end{changemargin}
	\caption{The relative depths and lengths of the reachability formulas over configurations of level $i$, and LTL formulas $\beta$ and $\tau$ of depth at most $d$ and length at most $l$. For the first two reachability formulas, we consider $i\geq 0$ and for the other formulas $i\geq 1$.}\label{tab:size}
\end{table}

\Subject{Depth Analysis}
The temporal nesting depth of the main reachability formula \smallerReach{$\Reach{S}{B}{\beta}{T}{\tau}$} is intuitively exponential in the number $n$ of levels of the reset cascade (linear in the number of configurations), since it is defined inductively along these levels, and the depth of a level-$(i+1)$ formula is about twice the depth of a level-$i$ formula. The parameters of the reachability formula are both the configurations $S$, $B$ and $T$ of level $i$, and the formulas  $\beta$ and $\tau$; yet, the depth of the reachability formula only linearly depends on the  depth of $\beta$ and $\tau$.
\Subject{Length Analysis}
Intuitively, the overall length of the main reachability formula \smallerReach{$\Reach{S}{B}{\beta}{T}{\tau}$} with respect to configurations of the top level is doubly exponential in the number $n$ of levels of the reset cascade (and thus singly exponential in the number of configurations), since the formula is defined inductively along these levels, and the length $\LF(i,l)$  is roughly $\LF(i{-}1,l) \cdot \LF(i{-}1,l)$. More precisely,  $\LF(i,l)=l \cdot f(i)$ for some doubly exponential function $f(i)$.

Now, why is $\LF(i,l)$ roughly equal to $\LF(i{-}1,l) \cdot \LF(i{-}1,l)$? The dominant component of the level-$i$ reachability formula is line (2) in the definition of \smallerReach{$\LeaveReach{\tuple{S,s}}{\tuple{B,b}}{\beta}{\tuple{T,t}}{\tau}$}.
It is a  level-$(i{-}1)$ reachability formula whose formula-parameters are themselves auxiliary reachability formulas of level $i$ with formula parameters of length $l$. The length of an auxiliary reachability formula of level $i$ is roughly as of the main reachability formula of level $i{-}1$, implying that the length of $L_i(l)$ is roughly $L_{i{-}1}(L_{i{-}1}(l))$. By the inductive proof that $L_{i{-}1}(l)=l \cdot f(i{-}1)$, we get that $L_i(l) = L_{i{-}1}(L_{i{-}1}(l)) = L_{i{-}1}(l) \cdot f(i{-}1) = l \cdot f(i{-}1) \cdot f(i{-}1)$.

As for the many disjunctions and conjunctions that appear in the formulas, observe that the number of disjuncts and conjuncts does not depend on the formula-parameters $\beta$ and $\tau$, but only the level $i$ of the configurations $S$, $B$, and $T$. Hence, they do not dominate the growth rate of the overall formula length.

\begin{restatable}{lemma}{ReachFormulaSize}\label{cl:ReachFormulaSize}
Consider a reset cascade $\A$ with $n$ levels and up to $n$ states in each level, and a formula \smallerReach{$\zeta = \Reach{S}{B}{\beta}{T}{\tau}$} with configurations $S$, $B$ and $T$ of $\A$ of level $i\leq n$. Let $d=\max(\depth{\beta}, \depth{\tau})$ and let $l=\max(|\beta|, |\tau|)$. Then: 
\[
(a)~\depth{\zeta} \leq d + 3^i \quad \text{ and } \quad (b)~|\zeta|\leq l \cdot  (10|\Sigma|^2n)^{4^i}
\]
\end{restatable}

\cref{cl:ReachFormulaSize} is proven by induction on $i$. The details are given in \cref{app:length}.
\subsection{Translating Deterministic Counter-Free Automata to LTL}\label{sec:DetAutomataToLTL}

We use the reachability formulas of \cref{sec:ReachbilityFormulas} to translate a reset cascade $\A$ to an equivalent LTL formula.
Our LTL formulation of $\A$'s acceptance condition is based on an LTL formulation of
``$C$ is visited finitely/infinitely often along a run of $\A$ on a word $w$'', for a given configuration $C$ of $\A$. It thus applies to every $\omega$-regular acceptance condition and by \cref{cl:ResetCascadeMuller,cl:KrohnRhodes} to every deterministic counter-free $\omega$-regular automaton.
We introduce two shorthands to the main reachability formula: the first is satisfied if we reach $T$ from $S$ without any side constraints (which is always satisfied in the case that $S = T$), and the second requires that we reach it along a nonempty prefix.  
\[\ReachTT{S}{T} \coloneqq \Reach{S}{T}{\false}{T}{\true} \quad\quad \NonEmptyReachTT{S}{T} \coloneqq \bigvee_{\sigma \in \Sigma} \left(\sigma \land \X (\ReachTT{\delta(S,\sigma)}{T})\right)\]%

With \cref{thm:reach-formulas-correctness,lem:reach-formulas-syntactic} we then obtain (a proof is given in \cref{app:main}):

\begin{restatable}{lemma}{LTLforVisitFinitelyOften}\label{cl:LTLforVisitFinitelyOften}
Consider a reset cascade $\A = \tuple{2^{\AP}, \A_1, \dots, \A_n}$ together with an initial configuration $\InitState$ and some configuration $C$. Then for a word $w\in {(2^{\AP})}^\omega$, the run of $\A$ on $w$ starting in $\InitState$ visits $C$ finitely often iff $w$ satisfies the formula
$\FIN(C) \coloneqq \neg(\ReachTT{\InitState}{C})\lor \ReachTT{\InitState}{C}(\neg (\NonEmptyReachTT{C}{C}))$. 
Furthermore, $\FIN(C) \in \Sigma_2$.
\end{restatable}

We are now in position to give our main result.

\begin{restatable}{theorem}{AutomatonToLTL}
\label{cl:AutomatonToLTL}
Every counter-free deterministic $\omega$-regular automaton $\D$ over alphabet $2^{\AP}$ with $n$ states (and any acceptance condition) is equivalent to an LTL formula $\varphi$ over atomic propositions $\AP$ of double-exponential temporal-nesting depth $($in $O(2^{2^n}))$ and triple-exponential length $($in $2^{2^{O(2^n)}})$. 
If $\D$ is a looping-B\"uchi, looping-coB\"uchi, weak, B\"uchi, coB\"uchi, or Muller automaton then $\varphi$ is respectively in the $\Pi_1, \Sigma_1, \Delta_1, \Pi_2, \Sigma_2$, or $\Delta_2$ syntactic fragment of LTL.
\end{restatable}

\begin{proof}
    We first prove the general result, w.r.t.\ an arbitrary counter-free deterministic automaton $\D$, and then take into account $\D$'s acceptance condition, to establish the last part of the theorem.

    Consider a counter-free deterministic $\omega$-regular automaton $\D$ with some acceptance condition  and $n$ states. Recall that there is a Muller automaton $\D'$ equivalent to $\D$ over the semiautomaton of $\D$.
	By \cref{cl:KrohnRhodes,cl:ResetCascadeMuller}, $\D'$ is equivalent to a deterministic Muller automaton $\D''$ that is described by a reset cascade $\A$ with up to $m=2^n$ levels and $m$ states in each level (and thus up to $m^m$ configurations), and whose acceptance condition has up to $k\in 2^{O(m^m n)} = 2^{O(m^m)}$ acceptance sets.
	An LTL formula $\phi$ equivalent to $\D$ can be defined by formulating the acceptance condition of $\D'$ along \cref{cl:LTLforVisitFinitelyOften}.
	
	Recall that the Muller condition is a $k$-elements disjunction, where each disjunct $M$ is a conjunction of requirements to visit infinitely often every configuration from some set $G$ and  finitely often every configuration not in $G$.
	Observe that $M$ can be formulated as a disjunction over all the configurations in $\D''$ (at most $m^m$), having for each configuration $C$ the LTL formula $\FIN(C)$ or $\neg\FIN(C)$, as defined in \cref{cl:LTLforVisitFinitelyOften}, depending on whether or not $C\in G$.
	Hence, the overall formula $\phi$ is a combination of disjunctions and conjunctions of up to $ k \cdot m^m$ subformulas of the form $\FIN(C)$ or $\neg\FIN(C)$.
	Therefore, the depth of $\phi$ is the same as of $\FIN(C)$, while $|\phi| \in O(k m^m |\FIN(C)|) \leq 2^{O(m^m)} |\FIN(C)|$.
	For calculating $\depth{\FIN(C)}$ and $|\FIN(C)|$, we use \cref{cl:ReachFormulaSize} bottom up over the subformulas of $\FIN(C)$.
	
	\Subject{Depth}\
	
	$\depth{\ReachTT{\InitState}{C}} \leq 3^m ~;~ \depth{\NonEmptyReachTT{C}{C}} \leq 3^m + 1$
	
	$\depth{\ReachTT{\InitState}{C}(\neg (\NonEmptyReachTT{C}{C}))} \leq 2 \cdot 3^m + 1$
	
	$\depth{\FIN(C)} = \max(3^m, 2\cdot 3^m + 1) \in O(3^m) = O(2^{2^n})$, 
	
	implying $\depth{\phi} \in O(2^{2^n})$.
	\Subject{Length}\
	
	$|\ReachTT{\InitState}{C}| \leq (10|\Sigma|^2m)^{4^m} ~;~ |\NonEmptyReachTT{C}{C}| \leq (4|\Sigma|) \cdot (10|\Sigma|^2m)^{4^m} $
	
	$|\ReachTT{\InitState}{C}(\neg (\NonEmptyReachTT{C}{C}))| \leq (4|\Sigma|(10|\Sigma|^2m)^{4^m} +1)(10|\Sigma|^2m)^{4^m} \in (|\Sigma| m)^{2^{O(m)}}$
	
	$|\FIN(C)| \in 2 + (10|\Sigma|^2m)^{4^m}+ (|\Sigma| m)^{2^{O(m)}} \in (|\Sigma| m)^{2^{O(m)}}$.
	
\hspace{0.1em}

	Therefore, $|\phi| \in 2^{O(m^m)} \cdot (m^m) \cdot ((|\Sigma| m)^{2^{O(m)}}) = |\Sigma|^{2^{O(m)}}$.
	
	Expressing the length of $\phi$ with respect to the number $n$ of states in the automaton $\D$, and taking into account the fact that the alphabet $\Sigma$ has at most $n^n$  different letters (any additional letter must have the same behavior as another letter), we have:
	$|\phi| \in |\Sigma|^{2^{O(2^n)}} \leq
	(2^n)^{2^{O(2^n)}} =
	2^{2^{O(2^n)}}$.
	
	We now sketch the second part of the theorem connecting the syntactic hierarchy and the different acceptance conditions of $\D$. We only consider the cases in which $\D$ is either a Muller or a coB\"uchi automaton. The complete analysis is given in \cref{app:main}. 
	If $\D$ is a Muller automaton, then the overall formula $\varphi$ is in $\Delta_2$, since it is a Boolean combination of $\FIN(C)$ formulas, which by \cref{cl:LTLforVisitFinitelyOften} belong to $\Sigma_2$.
	If $\D$ is a coB\"uchi automaton, then we construct the formula $\varphi$ directly from the coB\"uchi condition $\alpha$: $\varphi$ is a conjunction of $\FIN(C)$ formulas over all configurations $C$ that are mapped to states in $\alpha$. As $\FIN(C)$ belongs to $\Sigma_2$, so does $\varphi$.
\end{proof}

Observe that by \cref{cl:AutomatonToLTL}, we get the following result, extending the result of \cite[Theorem 3.2]{Tho81} that only considers Rabin automata.

\begin{corollary}
Every counter-free deterministic $\omega$-regular automaton (with any acceptance condition) recognises an LTL-definable language.
\end{corollary}
\begin{proof}
Recall that every deterministic $\omega$-regular automaton is equivalent to a deterministic Muller automaton over the same semiautomaton (see, e.g., \cite{Bok18}).
The claim is then a direct consequence of \cref{cl:AutomatonToLTL}.
\end{proof}

\begin{remark}\label{rem:FiniteWords}
\cref{cl:AutomatonToLTL} can be adapted to the finite-word setting. While on infinite words, the ne$\X$t operator is self-dual, i.e., $\neg \X \psi$ is equivalent to $\X \neg \psi$, over finite words, this equivalence does not hold on words of length 1. Thus $\X$ gains a dual \emph{weak next}, defined as $\tilde{\X} \psi \coloneqq \neg \X \neg \psi$. In the finite word case, syntactic cosafety (safety) formulas are constructed from $\true$, $\false$, $a$, $\neg a$, $\vee$, $\wedge$, and the temporal operators $\U$ and $\X$ ($\R$ and $\tilde{\X}$). Observe that $\X$ and $\tilde{\X}$ differ only on words of length 1, and thus the only required change in our translation scheme is to replace some $\X$s with $\tilde{\X}$s in the reachability formula 4.
For finite words a translation of a counter-free DFA to an LTL formula with only a double exponential size blow-up is known \cite{Wil99}; however, unlike our translation, it does not guarantee syntactic safety (cosafety) formulas for safety (cosafety) languages. 
\end{remark}

Lastly, we provide a corollary on looping automata, using  \Cref{cl:AutomatonToLTL} and the following known result.

\begin{proposition}[Rephrased Theorem 13 from \cite{DBLP:journals/ipl/MareticDB14}]
Let $\D$ be a deterministic looping-B\"uchi automaton with $n$ states that recognises an LTL-definable language. Then there exists an equivalent counter-free deterministic looping-B\"uchi automaton $\D'$ with at most $n$ states.
\end{proposition}

\begin{corollary}\label{cor:Looping}
Every deterministic looping-B\"uchi (looping-coB\"uchi) automaton with $n$ states that recognises an LTL-definable language is equivalent to an LTL formula $\varphi\in \Pi_1$ ($\Sigma_1$) of temporal nesting depth in $O(2^{2^n})$ and length in $2^{2^{O(2^n)}}.$
\end{corollary}

This is an elementary upper bound for two constructions for which either the upper bound was unknown or non-elementary: the liveness-safety decomposition of LTL \cite{DBLP:journals/ipl/MareticDB14} and the translation of semantic safety LTL to syntactic safety LTL.
\section{Conclusions}

We have studied the size trade-offs between LTL and automata. Over a unary alphabet, the situation is straightforward and we provided tight complexity bounds. The general case of infinite words over an arbitrary alphabet is more complex. We gave to our knowledge the first elementary complexity bound on the translation of counter-free deterministic $\omega$-regular automata into LTL formulas.

Every $\omega$-regular automaton recognising an LTL-definable language can be translated to a counter-free deterministic automaton \cite[Theorem 3.2]{Tho81}. Yet, we are unaware of a bound on the size blow-up involved in such a translation. Once established, it can be combined with our translation to get a general bound on the translation of automata to LTL. It will also provide a (currently unknown\footnote{In consultation with the author of \cite{Mar03}, we have confirmed that while the lower bound provided in that paper holds, the stated upper bound is erroneous.}) elementary upper bound on the translation of LTL with both future and past operators to LTL with only future operators (which is the version of LTL that we have considered), as (both version of) LTL can be translated to nondeterministic B\"uchi automata with a single exponential size blow-up \cite[Theorem 2.1]{VW86}.

While going from non-elementary to double-exponential depth and triple-ex\-ponential length is an improvement, these upper bounds might not be tight---there is currently no known non-linear lower bound! Closing this gap is a challenging open problem, which might require new lower bound techniques for alternating automata, as LTL formulas are an inherently alternating model.

\paragraph{Acknowledgements.}
We thank Moshe Vardi and Orna Kupferman for suggesting studying the succinctness gap between semantic and syntactic safe formulas, and Miko\l{}aj Boja\'nczyk for answering our questions on algebraic automata theory.

\bibliographystyle{splncs04}
\bibliography{bibliography.bib}


\vfill

{\small\medskip\noindent{\bf Open Access} This chapter is licensed under the terms of the Creative Commons\break Attribution 4.0 International License (\url{http://creativecommons.org/licenses/by/4.0/}), which permits use, sharing, adaptation, distribution and reproduction in any medium or format, as long as you give appropriate credit to the original author(s) and the source, provide a link to the Creative Commons license and indicate if changes were made.}

{\small \spaceskip .28em plus .1em minus .1em The images or other third party material in this chapter are included in the chapter's Creative Commons license, unless indicated otherwise in a credit line to the material.~If material is not included in the chapter's Creative Commons license and your intended\break use is not permitted by statutory regulation or exceeds the permitted use, you will need to obtain permission directly from the copyright holder.}

\medskip\noindent\includegraphics{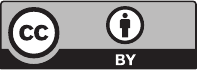}

\newpage
\appendix
\section{Omitted Proofs}
\label{appendix:proofs}

\subsection{Proofs from~\cref{sec:Unary}}
\label{appendix:unary}
\extendedWolper*

\begin{proof}
Consider an LTL formula $\phi$ and words $w_i = (u v^i t)$ and $w_j =(u v^j t)$, such that $i,j>\depth{\phi}$.
	We prove the claim by induction on the structure of $\phi$. 
	
	Base case: $\phi$ is an atomic proposition or a Boolean constant. Indeed, $\depth{\phi}=0$ and we have that $w_i \models \phi$ iff $w_j \models \phi$, because the first letter of these two words, which is the first letter of $u$ if $u$ is not empty and the first letter of $v$ otherwise, is the same.
	
	Induction step: We assume that the claim holds for all strict subformulas of $\phi$. Let $\psi_1$ and $\psi_2$ be strict subformulas of $\phi$. We show that the claim holds for:
	\begin{itemize}
		\item $\phi = \neg \psi_1$: Since $\depth{\phi}=\depth{\psi_1}$, it follows that $i,j>\depth{\psi_1}$, and therefore by the induction assumption $w_i\models \psi_1$ iff $w_j \models \psi_1$, implying that $w_i \models \phi$ iff $w_j \models \phi$.
		
		\item $\phi = \psi_1 \land \psi_2$:  Since $\depth{\phi}=\max(\depth{\psi_1}, \depth{\psi_2})$, it follows that $i,j>\depth{\psi_1}$ and $i,j>\depth{\psi_2}$. Therefore by the induction assumption ($w_i \models \psi_1$ iff $w_j \models \psi_1$) and ($w_i \models \psi_2$ iff $w_j \models \psi_2$). Hence, $w_i \models \psi_1 \land \psi_2$ iff $w_j \models \psi_1 \land \psi_2$.
		
		\item $\phi = \X \psi_1$: Recall that a word $w$ satisfies $\phi$ iff $w^1$ satisfies $\psi_1$.
		Observe that $w_i^1=u'v^{i-1}t$ and $w_j^1=u'v^{j-1}t$, where $u'=(uv)^1$. Since $\depth{\psi_1}=\depth{\phi}-1$, it follows that $i-1,j-1>\depth{\psi_1}$. Hence, by the induction assumption $w_i^1\models\psi_1$ iff $w_j^1\models\psi_1$, and therefore $w_i\models\phi$ iff $w_j\models\phi$.

		\item $\phi = \psi_1 \U \psi_2$: We will show that if $w_i$ satisfies $\phi$ then so does $w_j$.
		
		If $w_i$ satisfies $\phi$ then there is a position $p$ of $w_i$, such that $w_i^p\models\psi_2$ and for every $k<p$, $w_i^k\models\psi_1$.
		Let $o$  be the position of $w_i$ that appears at the beginning of the $v$-block that is $\depth{\phi}$ blocks of $v$ before the $t$ part of $w_i$, namely $o=|uv^{i-\depth{\phi}}|$.	(See Figure~\ref{fig:w_i}.)
		\begin{figure}
			\centering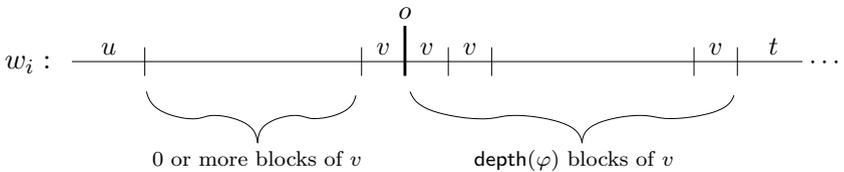 \caption{The structure of the word $w_i$ from the proof of \cref{cl:Nesting}.}\label{fig:w_i}
		\end{figure}

		We split the proof into disjoint cases, depending on the location of $p$ within $w_i$.
		\begin{itemize}
			\item $p<|u|$: Let $u'$ be the infix of $w_j$ from $p$ to the end of $u$, namely $u'=w_i[p..|u|-1]$. Then $w_i^p=u'v^it$ and $w_j^p=u'v^jt$. Since $i,j>\depth{\phi}>\depth{\psi_2}$, by the induction assumption $w_i^p\satisfies \psi_2$ iff $w_j^p\satisfies \psi_2$ and therefore $w_j^p\satisfies \psi_2$. Likewise, since $i,j>\depth{\phi}>\depth{\psi_1}$, by the induction assumption for every position $m<p$, $w_i^m\satisfies \psi_1$ iff $w_j^m\satisfies \psi_1$ and therefore $w_j^m\satisfies \psi_1$.

			\item $p\in[|u|..o-1]$:	Let $p'$ be the position in $w_j$ that appears in the first $v$-block after $u$ and that is located within that $v$-block like $p$ is located within its $v$-block. That is, $p' = |u| + ((p-|u|) \mod |v|)$. Let $u'$ be the remaining suffix in the $v$-block of $p$ and $p'$, that is $u'=w_j[p'+1..|u|+|v|]$. 
			 Let $h$ be the number of $v$-blocks that appear after $p$ and before the $t$ part of $w_i$, that is $h=\depth{\phi} + \lfloor (o - p) / |v| \rfloor$.
			 Then $w_i^{p}=u'v^{h}t$ and $w_j^{p'}=u'v^{j-1}t$. Since $h\geq\depth{\phi}>\depth{\psi_2}$ and $j-1\geq\depth{\phi}>\depth{\psi_2}$, we have by the induction assumption that $w_j^{p'}\satisfies \psi_2$ iff $w_i^{p}\satisfies \psi_2$, and therefore $w_j^{p'}\satisfies \psi_2$.
			
			Now, for every position $m<p'$, let $u'$ be the infix of $w_j$ from $m$ to the end of the first $v$-block of $w_j$, that is $u'=w_j[m..|u|+|v|]$. Then $w_i^m=u' v^{i-1} t$ and $w_j^m=u' v^{j-1} t$. Since $i-1>\depth{\psi_1}$ and $j-1>\depth{\psi_1}$, by the induction assumption $w_i^m\satisfies \psi_1$ iff $w_j^m\satisfies \psi_1$ and therefore $w_j^m\satisfies \psi_1$.
			
			\item $p\geq o$: Let $p'$ and $o'$ be the positions in $w_j$ that are at the same distance from the $t$ part of $w_j$ as $p$ and $o$ are from the $t$ part $w_i$, namely $o'=|uv^{j-\depth{\phi}}|$ and $p'=o' + (p-o)$. Observe that $w_i^o=w_j^{o'}$ and $w_j^{p'} = w_i^{p}$, implying that $w_j^{p'}\satisfies \psi_2$. 
			
			Further, for every position $m\in[o-|v|..p]$ of $w_i$, let $m'$ be the corresponding position in $w_j$, namely $m'=o' + (m-o)$. Then $w_j^{m'} = w_i^{m}$, and accordingly $w_j^{m'}\satisfies \psi_1$.
			
			Now, for every position $m'\in[|u|..o-|v|-1]$, we have by the induction assumption that $w_j^{m'}\satisfies \psi_1$ iff $w_j^{m'+|v|}\satisfies \psi_1$ (as both words have the same prefix until the end of the first $v$ block, followed by at least $\depth{\psi_1}+1$ blocks of $v$ and then $t$), implying that $w_j^{m'}\satisfies \psi_1$.
			
			Finally, for every position $m\in[0..|u|-1]$, we have by the induction assumption that $w_i^m\satisfies \psi_1$ iff $w_j^{m}\satisfies \psi_1$, implying that $w_j^{m}\satisfies \psi_1$.
	\end{itemize}
\end{itemize}
\end{proof}

\FiniteCofinite*
\begin{proof}
	Let $\# = \{p\}$ be the shorthand for a new atomic proposition $p$.
	Given a unary LTL formula recognising  $L \subseteq \{a\}^+$, it can be turned into a general LTL formula of linear length that recognises $L'=\{v\#^\omega\mid v\in L\}$ over the alphabet $\Sigma' = 2^{\{p\}}$. Then, the statement is a direct consequence of~\cref{cl:Nesting}.
\end{proof}

\LTLforShortLang*
\begin{proof}
	The very weak deterministic automaton for $L$ consists of $n+2$ states $\{0,\dots, n+1\}$, with an $a$-transition from state $i$ to state $i+1$ and a self loop at state $n+1$. The state $i$ is accepting whenever $a^i\in L$. 
\end{proof}

\UnaryBlowup*
\begin{proof}
Deterministic automata: For the upper bound, consider a DFA $\A$ recognising an LTL definable language.
By a pumping argument, if $w\in L(\A)$ for some $w$ longer than $n$, then $L(\A)$ is infinite. Considering the dual automaton $\A'$ of $\A$, recognising the complement language, by the pumping argument if $w\not\in L(\A)$ for some $w$ longer than $n$, then the complement of $L(\A)$ is infinite. Hence, by \cref{cl:finite-cofinite}, $L(\A)$ agrees on all words of length more than $n$, and by \cref{cl:LTLforShortLang} there is an LTL formula for it of length in $O(n)$.
As for the lower bound, since there is a DFA of size $n$ recognising ${a^k}$, it directly follows from \cref{cl:Nesting}.

Nondeterministic automata: Directly follows from \cref{cl:UnaryAltUpper,cl:AlternatingCounter}.

Alternating automata: Directly follows from \cref{cl:NFAtoLTL,cl:NFAforSquare}.
\end{proof}
\UnaryAltUpper*

\begin{proof}
First recall that the run of an alternating automaton is a tree, of which all paths are accepting if and only if the run itself is accepting. These runs can be pumped in the same way as runs of nondeterministic automata, except that the run to be pumped must be of length over $2^n$ to guarantee that the \textit{set of states} in a cross-section of the run are repeated.

Then, by a pumping argument, if $w\in L(\A)$ for some $w$ longer than $2^n$, then $L(\A)$ is infinite: indeed, let $\A'$ be an NFA equivalent to $\A$ of size at most $2^n$; if $\A'$ accepts a word longer than $2^n$, then its run sees some state more than once and can be pumped to build infinitely many accepting runs. Dually, since AFA are easy to complement, if $w'\notin L(\A)$ for some $w'$ longer than $2^n$, then $\Sigma^*\setminus L(\A)$ is infinite. The existence of both $w\in L(\A)$ and $w'\notin L(\A)$ longer than $2^n$ therefore contradicts Lemma \ref{cl:finite-cofinite}. We conclude that $L(A)$ agrees on all words of length over $2^n$. Thus, by \cref{cl:LTLforShortLang} there is an LTL formula for $L(\A)$ of length in $O(2^n)$.
\end{proof}

\AlternatingCounter*
\begin{proof}
	The idea of the construction is that it represents an $n$-bit up-counter, having two states for each bit, one corresponding to $1$, one to $0$ (see \cref{fig:AltCounter}).
	
	Given a way to resolve the nondeterministic choices, the resulting set of ``active'' states of the automaton represents a configuration of the counter: The nondeterminism in each bit-state chooses whether to change the bit's value (going left in \cref{fig:AltCounter}), in which case the universality ensures that all lower bits are set to $0$, or to preserve the bit's value (going right in \cref{fig:AltCounter}), in which case the universality ensures that at least one lower bit is set to $1$.
	
	The automaton thus preserves the invariant that a correct update (for example from $011$ to $100$) ensures that the number of active states is constant, in particular, if all updates are correct, then exactly one state per bit is active at a time; an incorrect update on the other hand increases the number of active states. The set of accepting states corresponds to the bit-configuration for $2^{n-1}$, where a transition to a state $q_{i\ZB}$, for every $i<n$, can be changed to a move to $q_{acc}$, provided that it is on the last letter of the input word. (There is no transition out of $q_{acc}$). The only way to build an accepting run is to correctly update the counter at each step; an incorrect update ensures that both states of some bit are set from thereon, of which one must be rejecting.
	
	As for the automaton size, which is the number of subformulas in the transition function, observe that it is linear in $n$, since there are $2n$ states and the total number of subformulas in the transition function is as follows: The transition functions of  $q_{1\ZB}$ and $q_{1\OB}$ have together four subformulas, and the transition function of  every other state $q_{(i,\cdot)}$ adds up to $4$ subformulas: (i) the topmost disjunction of whether to change the bit's value or not; (ii)\&(iii) the universality involved in each of these two options; and (iv) within each of the two universality subformulas -- a disjunction or conjunction between  $q_{(i-1,\cdot)}$ and $(q_{(i-1,\cdot)},\ldots,q_{(i-1,\cdot)})$, where the latter subformula (within the parenthesis) is already used in the transition function of $q_{(i-1,\cdot)}$ and these two formulas (one of conjunction and one of disjunction) is used by both $q_{i\ZB}$ and $q_{i\OB}$ (except for $q_{n\OB} = q_{acc}$, which has no outgoing transitions).
\end{proof}

\NFAtoLTL*

\begin{proof}
If $L(\A)$ is finite, then by a pumping argument, $\A$ only accepts words up to length $n$, and by \cref{cl:LTLforShortLang} we are done.
We consider co-finite $L(\A)$. 

We will argue using 2-way deterministic automata, which are deterministic automata that process words of the form $\vdash\!\!w\!\!\dashv$, where $\vdash$ and $\dashv$ are start- and end-of-word markers respectively, and where transitions specify whether to read the letter to the right or to the left of the current position. They accept by reaching an end state, and reject by reaching a rejecting state or by failing to terminate. For a more formal definition, see~\cite{GMP07}. We will use the facts that a unary NFA can be turned into a 2-way DFA of size $O(n^2)$~\cite{Chr86}, that a 2-way DFA of size $m$ can be complemented into one of size $4m$ \cite{GMP07}, and that the smallest 2-way DFA recognising $\{a^k\}$ is of length $k+2$~\cite{Bir96}.

The NFA $\A$ can be turned into a 2-way DFA $\D$ with $O(n^2)$ states recognising the same language. From $\D$, we can obtain a 2-way DFA $\D'$ that recognises  $a^*\setminus \{a^k\}$, where $a^k$ is the longest word not in $L(\A)$, as follows: we keep the states and transitions involved in the run of $\D$ on $\vdash\!\!a^k\!\!\dashv$, which is rejecting and must read $\dashv$ (since, by the co-finiteness of $L(\A)$, there are accepted words longer than $a^k$). Notice that states in $\D'$ may have only some of the transitions they had in $\D$ -- only those that are involved in the run on $a^k$.  We then add an $a$-transition and a $\dashv$-transition from all states without such transitions, and these all lead to the (accepting) end state. 

Notice that $\D'$ is still deterministic and not larger than $\D$. It still rejects $a^k$, on which it has the same run as $\D$, but accepts all other words, which are either shorter than $a^k$ and therefore accepted via one of the new $\dashv$-transitions, or are longer than $a^k$ and accepted when they read the $k+1^{th}$ letter.

We can then complement $\D'$ into a 2-way DFA $\D''$ of size linear in the size of $D'$, namely in $O(n^2)$, that recognises $\{a^k\}$. However, since $\D''$ must be of size at least $k+2$, 
we get that $k+2$ is at most in $O(n^2)$, meaning that the longest word rejected by $\A$, which is of length $k$, is of length in $O(n^2)$.

Then, there is an equivalent very weak automaton of the size of the longest word not in $L(\A)$, from \cref{cl:AlternatingCounter}. 
\end{proof}

\subsection{Proofs from~\cref{sec:CascadedAutomata}}

\RabinCascade*

\begin{proof}
	We provide the proof for Rabin automata. The proofs for B\"uchi and coB\"uchi automata are special cases of the Rabin case.
	
	Let $\D=(\Sigma,Q,\iota,\delta,\alpha)$, and $h$ be the mapping via which $(\Sigma, Q, \delta)$ is homomorphic to $\A$.
	We shall define a configuration $\iota'$ of $\A$, and a set $\alpha'$ of pairs of sets of configurations of $\A$ that will provide the initial state and acceptance condition of $\D'$, making it equivalent to $\D$.
	
	For the initial state of $\D'$, one can choose any configuration $\iota'\in h^{-1}(\iota)$: A run of $\A$ starting in $\iota'$ corresponds via $h$ to a run of $\D$ starting in $h(\iota')=\iota$.
	
	As for the acceptance condition, consider the run $r'$ of $\D'$ on a word $w$, starting in $\iota'$, and let $r=h(r')$ be the corresponding run of $\D$ on $w$. (With $h(r')$, we mean the sequence of states of $\D$, obtained by mapping  with $h$ each configuration of $r'$.)
	Then, $r'$ should be accepting iff $r$ is.
	
	Recall that $r$ is accepting iff $\Inf(r)\cap G \neq \emptyset$ and $\Inf(r)\cap B = \emptyset$ for some pair $(G,B)\in\alpha$.
	Hence, $r'$ should be accepting ``according to $(G,B)$'' iff $h(\Inf(r'))\cap G \neq \emptyset$ and $h(\Inf(r'))\cap B = \emptyset$.
	
	Thus, $r'$ should visit infinitely often some configuration in $G'=h^{-1}(G)$ and finitely often every configuration in $B'=h^{-1}(B)$.
	Hence, there are $k$ acceptance pairs $(G',B')$ in $\alpha'$, each corresponding to an acceptance pair $(G,B)$ in $\D$.
\end{proof}

\MullerCascade*

\begin{proof}
	Let $\D=(\Sigma,Q,\iota,\delta,\alpha)$, and $h$ be the mapping via which $(\Sigma, Q, \delta)$ is homomorphic to $\A$.
	We shall define a configuration $\iota'$ of $\A$, and a set $\alpha'$ of sets of configurations of $\A$ that will provide the initial state and acceptance condition of $\D'$, making it equivalent to $\D$.
	
	For the initial state of $\D'$, one can choose any configuration $\iota'\in h^{-1}(\iota)$: A run of $\A$ starting in $\iota'$ corresponds via $h$ to a run in $\D$ starting in $h(\iota')=\iota$.
	
	As for the acceptance condition, consider the run $r'$ of $\D'$ on $w$, starting in $\iota'$, and let $r=h(r')$ be the run of $\D$ on $w$ (that is defined by mapping each configuration of $r'$ to a state of $\D$ via $h$).
	Then, $r'$ should be accepting iff $r$ is.
	
	Recall that $r$ is accepting iff $\Inf(r)=M$, for some $M=\{q_1, \ldots q_l\}\in\alpha$. Hence, $r'$ should be accepting ``according to $M$'' iff $h(\Inf(r'))=M$.
	Thus, $r'$ should visit finitely often every configuration in $h^{-1}(Q\setminus M)$, and for \emph{every} $i\in[1..l]$, visit \emph{some} configurations in $G_i=h^{-1}(q_i)$ infinitely often.
	
	Since we should consider every choice of configurations in $G_i=h^{-1}(q_i)$, where $|G_i|\leq m$, there are up to $2^m$ choices for $G_i$, and therefore up to $(2^m)^n$ choices for $M$, each providing a Muller set $G$ of configurations to be visited infinitely often.
	
	As there are up to $2^n$ sets in $\alpha$, we end up with up to $2^n (2^m)^n \in 2^{O(mn)}$ sets in $\alpha'$.
\end{proof}

\subsection{Proofs from~\cref{sec:ReachbilityFormulas}}
\label{app:correctness}

\thmCorrectness*

\begin{proof}
Observe first that there is no circularity in the definitions of the five reachability formulas, even though they are defined by each other: Formula 2 is defined on top of formula 1, which is defined on top of formulas 3 and 5, while formulas 3, 4, and 5 are defined with respect to reachability formulas over configurations of a lower level.

We prove the statement by induction on the level $m$ of the configurations $S$, $B$, $T$ in the reachability formulas. We split the proof to five cases corresponding to the five reachability formulas in \Cref{table:reach-formulas}. In order to clarify which equivalence of the induction hypothesis is used we denote by (I.H.1) the equivalence in \Cref{table:reach-formulas} for reachability formula 1, (I.H.2) the equivalence for reachability formula 2, and so on.

\bigskip
\noindent \textbf{Reachability formula 1}

\noindent ($m = 0$): There is exactly one configuration of level $0$ and it is the empty configuration. Thus $S = T = B = \tuple{}$. We then derive:
\[\begin{array}{cl}
     & w \models \Reach{\tuple{}}{\tuple{}}{\beta}{\tuple{}}{\tau}  \\[0.9em]
\iff & w \models (\neg \beta) \U \tau \\[0.6em]
\iff & \exists i \geq 0. ~ w_{[i..]} \models \tau \land \forall j \in [0..i). ~ w_{[j..]} \not \models \beta \\[0.6em]
\iff & \exists i \geq 0. ~ \delta(\tuple{}, w_{[0..i)}) = \tuple{} \land w_{[i..]} \models \tau \\
     & \qquad \qquad \land ~ \forall j \in [0..i). ~ \delta(\tuple{}, w_{[0..j)}) \neq \tuple{} \lor w_{[j..]} \not \models \beta
\end{array}\]

\noindent ($m \rightarrow m + 1$): Let $\tuple{S,s}, \tuple{T,t}, \tuple{B,b} \in Q_1 \times \cdots \times Q_{m+1}$ be arbitrary configurations of level $m+1$. Thus we can use the equalities from \Cref{table:reach-formulas} for all configurations of $Q_1 \times \cdots \times Q_m$ as induction hypotheses. We need to show that:
\[\begin{array}{cl}
     & w \models \Reach{\tuple{S,s}}{\tuple{B,b}}{\beta}{\tuple{T,t}}{\tau} \\[0.9em]
\iff & w \models \StayReach{\tuple{S,s}}{\tuple{B,b}}{\beta}{\tuple{T,t}}{\tau} \lor \LeaveReach{\tuple{S,s}}{\tuple{B,b}}{\beta}{\tuple{T,t}}{\tau} \\[0.9em]
\iff & \exists i \geq 0. ~ \delta(\tuple{S,s},w_{[0..i)}) = \tuple{T,t} \land w_{[i..]} \models \tau \\ 
  	 & \qquad \qquad \land ~ (\forall j \in [0..i). ~ \delta(\tuple{S,s},w_{[0..j)}) \neq \tuple{B,b} \lor w_{[j..]} \not \models \beta)
\end{array}\]  

\noindent We split this into the ($\Rightarrow$)- and ($\Leftarrow$)-direction:

\medskip
\noindent ($\Rightarrow$): We further refine this and first assume that the first disjunct is satisfied by $w$ and defer the other case to a later point in the proof. We then apply (I.H.3) and derive:

\[\begin{array}{cl}
     & w \models \StayReach{\tuple{S,s}}{\tuple{B,b}}{\beta}{\tuple{T,t}}{\tau} \\[0.9em]
\iff & \exists i \geq 0. ~ \delta(\tuple{S,s},w_{[0..i)}) = \tuple{T,t} \land w_{[i..]} \models \tau \\ 
  	& \qquad \qquad \land ~ (\forall j \in [0..i). ~ \delta(\tuple{S,s},w_{[0..j)}) \neq \tuple{B,b} \lor w_{[j..]} \not \models \beta) \\
	& \qquad \qquad \land ~ (\forall j \in [0..i). ~ \tuple{w[j], \delta(S, w_{[0..j)})} \in \Stay{s}) \\
\implies & \exists i \geq 0. ~ \delta(\tuple{S,s},w_{[0..i)}) = \tuple{T,t} \land w_{[i..]} \models \tau \\ 
  	 & \qquad \qquad \land ~ (\forall j \in [0..i). ~ \delta(\tuple{S,s},w_{[0..j)}) \neq \tuple{B,b} \lor w_{[j..]} \not \models \beta)
\end{array}\]  

We now assume that the first disjunct is \emph{not} satisfied by $w$ and thus the second disjunct is satisfied by $w$. We then derive using (I.H.5):

\[\begin{array}{cl}
     & w \models \LeaveReach{\tuple{S,s}}{\tuple{B,b}}{\beta}{\tuple{T,t}}{\tau} \\[0.9em]
\iff & \exists i_1, i_2 \geq 0. ~ \delta(\tuple{S,s},w_{[0..i_1)}) = \tuple{T,t} \land w_{[i_1..]} \models \tau \\ 
    & \qquad \qquad \land ~ (\exists j_1 \in [0..i_1). ~ \tuple{w[j_1], \delta(S,w_{[0..j_1)})} \in \Enter{t} \\ 
    & \qquad \qquad \land ~ \tuple{w[i_2], \delta(S,w_{[0..i_2)})} \in \Leave{s} \\ 
    & \qquad \qquad \land ~ (\forall j \in [0..\max(i_1{-}1, i_2)]. ~ \delta(\tuple{S,s},w_{[0..j)}) \neq \tuple{B,b} \lor w_{[j..]} \not \models \beta) \\
\implies & \exists i \geq 0. ~ \delta(\tuple{S,s},w_{[0..i)}) = \tuple{T,t} \land w_{[i..]} \models \tau \\ 
  	 & \qquad \qquad \land ~ (\forall j \in [0..i). ~ \delta(\tuple{S,s},w_{[0..j)}) \neq \tuple{B,b} \lor w_{[j..]} \not \models \beta)
\end{array}\] 

\medskip
\noindent ($\Leftarrow$): We assume that $w$ satisfies the right-hand side of the equation and we instantiate $i$ to be the smallest non-negative integer such that:
\[\begin{array}{cl}
& \delta(\tuple{S,s},w_{[0..i)}) = \tuple{T,t} \land w_{[i..]} \models \tau \\
& \qquad \qquad \land ~ (\forall j \in [0..i). ~ \delta(\tuple{S,s},w_{[0..j)}) \neq \tuple{B,b} \lor w_{[j..]} \not \models \beta)
\end{array}\] 

Assume that for all proper prefixes of $w_{[0..i)}$ the combined letter stays in $s$, i.e., $\tuple{w[j], \delta(S, w_{[0..j)})} \in \Stay{s}$ for all $j \in [0..i)$. We then apply (I.H.3) and obtain that $w$ satisfies \smallerReach{$\StayReach{\tuple{S,s}}{\tuple{B,b}}{\beta}{\tuple{T,t}}{\tau}$} and so the left-hand side of the equation is satisfied by $w$. Thus let $k \in [0..i)$ be the smallest non-negative integer such that:
\[\tuple{w[k], \delta(S, w_{[0..k)})} \in \Leave{s}\]
Since $k < i$, we immediately obtain one of the missing preconditions for applying (I.H.5): \[\forall j_2 \in [0..k]. ~ \delta(\tuple{S,s},w_{[0..j_2)}) \neq \tuple{B,b} \lor w_{[j_2..]} \not \models \beta\]
The second missing precondition is that we need to find a $j_1 \in [0..i)$ such that $\tuple{w[j_1], \delta(S, w_{[0..j_1)})} \in \Enter{t}$:
In the case $s \neq t$ this is straightforward, since after reading $w_{[0..i)}$ we reach $\tuple{T,t}$, and thus there must be a $j_1 < i$ such that $\tuple{w[j_1], \delta(S, w_{[0..j_1)})} \in \Enter{t}$. 
Thus let us assume $s = t$. In general there might be not such an index $j_1$. However, we have shown that $\tuple{w[k], \delta(S, w_{[0..k)})} \in \Leave{s}$ for some $k < i$ and by the same reasoning as above there must be $j_1$ between $k$ and $i$ such that $\tuple{w[j_1], \delta(S, w_{[0..j_1)})} \in \Enter{t}$.

We now can apply (I.H.5) and conclude this direction of the proof.

\bigskip
\noindent \textbf{Reachability formula 2}

\noindent We proceed by a straightforward derivation for which we use (I.H.1) in the second step:
\[\begin{array}{cl}
     & w \models \WeakReach{S}{B}{\beta}{T}{\tau} \\[0.9em]
\iff & w \not \models \Reach{S}{T}{\tau}{B}{\beta} \\[0.9em]
\iff & \neg \big(\exists i \geq 0. ~ \delta(S,w_{[0..i)}) = B \land w_{[i..]} \models \beta \\ 
  	 & \qquad \qquad \land ~ (\forall j \in [0..i). ~ \delta(S,w_{[0..j)}) \neq T \lor w_{[j..]} \not \models \tau)\big)\\
\iff & \forall i \geq 0. ~ (\delta(S,w_{[0..i)}) = B \land w_{[i..]} \models \beta )\\ 
  	 & \qquad \qquad \rightarrow ~ (\exists j \in [0..i). ~ \delta(S,w_{[0..j)}) = T \land w_{[j..]} \models \tau)
\end{array}\]  

\bigskip
\noindent \textbf{Reachability formula 3}

\noindent We want to prove the following equivalence:  

\[\begin{array}{cl}
     & w \models \StayReach{\tuple{S,s}}{\tuple{B,b}}{\beta}{\tuple{T,t}}{\tau} \\
\iff &         \exists i \geq 0. ~ \delta(\tuple{S,s},w_{[0..i)}) = \tuple{T,t} \land w_{[i..]} \models \tau \\ 
  	 &	       \qquad \qquad	 \land ~ (\forall j_1 \in [0..i). ~ \delta(\tuple{S,s},w_{[0..j_1)}) \neq \tuple{B,b} \lor w_{[j_1..]} \not \models \beta) \\
	 & 		   \qquad \qquad	 \land ~ (\forall j_2 \in [0..i). ~ \tuple{w[j_2],\delta(S,w_{[0..j_2)})} \in \Stay{s})
\end{array}\]

\noindent ($\Rightarrow$): Assume that $w$ satisfies the left-hand side of the equivalence. Notice that if $w$ does not satisfy $\varphi \coloneqq \StayReachNonEmpty{\tuple{S,s}}{\tuple{B,b}}{\beta}{\tuple{T,t}}{\tau}$, then necessarily $\delta(\tuple{S,s}, w_{[0..0)}) = \tuple{T,t}$ and $w_{[0..]} \models \tau$ and thus the right-hand side trivially holds for $i = 0$. 

Thus we can assume that $w$ satisfies $\varphi$ and using the same reasoning we know that either $\delta(\tuple{S,s},w_{[0..0)}) = \tuple{B,b}$ or $w_{[0..]} \models \beta$ does not hold, which takes care of the case $j_1 = 0$ in the second line of the right-hand side, assuming $i > 0$.  

Since we have $w \models \varphi$, there must be a $\tuple{\sigma, T'} \in \Stay{s}$ with $\delta(\tuple{T',s}, \sigma) = \tuple{T,t}$ such that the matching disjunct $\psi$ of $\varphi$ is satisfied by $w$. Note that this immediately implies that $s=t$. Observe that $\psi$ is a conjunction of formulas with the shape \[\Reach{S}{\cdot}{\cdot}{T'}{\sigma \land \X \tau}\] and we now apply (I.H.1) to all reachability formulas. Since they all share the same target, we can instantiate them to the same $i'$ (where $i'$ is the smallest non-negative integer satisfying the conditions) such that:

\begin{enumerate}[label=(\alph*)]
	\item $\delta(S, w_{[0..i')}) = T'$
	\item $w[i'] = \sigma$
	\item $w_{[i'+1..]} \models \tau$
	\item For every $\tuple{\eta, L} \in \Leave{s}$ and every $j \in [0..i')$ at least one of the following statements holds:
	\begin{enumerate}[label=(\roman*)]
		\item $\delta(S, w_{[0..j)}) \neq L$
		\item $w[j] \neq \eta$
	\end{enumerate}
	\item For every $\tuple{\rho, B'} \in \Stay{s}$ such that $\delta(\tuple{B',s},\rho) = \tuple{B,b}$  and every $j \in [0..i')$ at least one of the following statements holds:
	\begin{enumerate}[label=(\roman*)]
	    \item $\delta(S, w_{[0..j)}) \neq B'$
	    \item $w[j] \neq \rho$
	    \item $w_{[j+1..]} \not \models \beta$
	\end{enumerate}
\end{enumerate}

We now instantiate $i$ of the right-hand side with $i'+1$. Remember that we have $\tuple{\sigma, T'} \in \Stay{s}$ and together with (a-c) we obtain the first conjunct of the right hand side. 

We now establish the third conjunct. Note that for every $j_2 \in [0..i')$ we have $\tuple{w[j_2], \delta(\tuple{S,s}, w_{[0..j_2)})} \in \Stay{s}$, since $\Stay{s} = (\Sigma \times Q_1 \times \cdots \times Q_m) \setminus \Leave{s}$ and by (d) we do not encounter an element of $\Leave{s}$ for any $j_2$. Further, for $j_2 = i'$ we obtain $\tuple{w[j_2], \delta(\tuple{S,s}, w_{[0..j_2)})} \in \Stay{s}$ from (a,b) and the choice of $\tuple{\sigma, T'} \in \Stay{s}$.

For the second conjunct it remains to show that for every $j_1 \in [0..i)$:
\[\delta(\tuple{S,s},w_{[0..j_1)}) \neq \tuple{B,b} \lor w_{[j_1..]} \not \models \beta\]
Assume that $\delta(\tuple{S,s},w_{[0..j_1)}) = \tuple{B,b}$ (if this is not the case we are immediately done). Due to the already established third conjunct, we have $\tuple{w[j_1-1], \delta(S, w_{[0..j_1-1)})} = \tuple{\rho, B'} \in \Stay{s}$ and $\delta(\tuple{B',s}, \rho) = \tuple{B,b}$. Thus (e,i) and (e,ii) cannot hold and (e,iii) must hold for $j_1 (= j + 1)$.

\medskip
\noindent ($\Leftarrow$): We assume that $w$ satisfies the right-hand side and instantiate $i$ as the smallest $i\geq0$ such that:
\begin{enumerate}[label=(\alph*)]
	\item $\delta(\tuple{S,s},w_{[0..i)}) = \tuple{T,t}$
	\item $w_{[i..]} \models \tau$
	\item $\forall j \in [0..i). ~ \delta(\tuple{S,s},w_{[0..j)}) \neq \tuple{B,b} \lor w_{[j..]} \not \models \beta$
	\item $\forall j \in [0..i). ~ \tuple{w[j],\delta(S,w_{[0..j)})} \in \Stay{s}$
\end{enumerate}

If $i = 0$, then due to (a-b) we have immediately $w \models \StayReach{\tuple{S,s}}{\tuple{B,b}}{\beta}{\tuple{T,t}}{\tau}$. Thus without loss of generality we can assume $i > 0$ from now on. Further, due to (c) we have that either $\tuple{S,s} \neq \tuple{B,b}$ or $w \not \models \beta$. Thus it remains to show that $w \models \StayReachNonEmpty{\tuple{S,s}}{\tuple{B,b}}{\beta}{\tuple{T,t}}{\tau}$.
Define $\sigma$ as $w[i-1]$ and $T'$ as $\delta(S,w_{[0..i-1)})$. From (a,d) we then follow:

\begin{enumerate}
	\item[(e)] $\tuple{\sigma,T'} \in \Stay{s}$
	\item[(f)] $\delta(S,w_{[0..i-1)}) = T'$ and $\delta(\tuple{T',s}, \sigma) = \tuple{T,t}$
	\item[(g)] $w_{[i-1..]} \models \sigma \land \X \tau$
\end{enumerate}

We are going to use the induction hypothesis to show that $w$ satisfies the disjunct corresponding to $\tuple{\sigma,T'}$. The first conjunct is satisfied from (f,g) and (I.H.1). From (d) it follows that for every $j \in [0..i)$ the tuple $\tuple{w[j],\delta(S,w_{[0..j)})}$ is not in $\Leave{s}$. Thus for every $\tuple{\eta, L} \in \Leave{s}$, we either have $w_{[j..]} \not \models \eta$ or $\delta(S, w_{[0..j)}) \neq L$. Thus by (I.H.1) and (f,g) we obtain that $w$ satisfies the second conjunct. Analogously, from (c) it follows that for every $j \in [1..i)$ and for every $\tuple{\rho, B'} \in \Stay{s}$ such that $\delta(\tuple{B',s}, \rho) = \tuple{B,b}$, either $\delta(S,w_{[0..j-1)}) \neq B'$ or $w_{[j-1..]} \not \models \rho \land \X \beta$. Thus by (I.H.1) and (f,g) we obtain that $w$ satisfies the third and final conjunct.

\bigskip
\noindent \textbf{Reachability formula 4}

\noindent We want to prove the following equivalence:  

\[\begin{array}{cl}
     & w \models \WeakStayReach{\tuple{S,s}}{\tuple{B,b}}{\beta}{\tuple{T,t}}{\tau} \\[1.8em]
\iff & \forall i \geq 0. ~\big((\delta(\tuple{S,s},w_{[0..i)}) = \tuple{B,b} \land w_{[i..]} \models \beta)  \\ 
    &  \phantom{\forall i \geq 0. ~} \qquad \lor (i>0 \wedge \tuple{w[i{-}1], \delta(S,w_{[0..i{-}1)})} \in \Leave{s}) \big) \\ 
  	&  \phantom{\forall i \geq 0. ~} \rightarrow ~ (\exists j \in [0..i). ~ \delta(\tuple{S,s},w_{[0..j)}) = \tuple{T,t} \land w_{[j..]} \models \tau) 
\end{array}\]

\noindent ($\Rightarrow$): Assume that $w$ satisfies the left-hand side of the equivalence. Then we immediately obtain from the definition that either $\delta(\tuple{S,s}, w_{[0..0)}) \neq \tuple{B,b}$ or $w_{[0..]} \not \models \beta$. Further, notice that if $w$ does not satisfy $\varphi \coloneqq \WeakStayReachNonEmpty{\tuple{S,s}}{\tuple{B,b}}{\beta}{\tuple{T,t}}{\tau}$, then necessarily $\delta(\tuple{S,s}, w_{[0..0)}) = \tuple{T,t}$ and $w_{[0..]} \models \tau$, and thus the right-hand side trivially holds for all $i$ by instantiating $j$ with $0$. 

Thus we can assume that $w$ satisfies $\varphi$ and using the fact that at least one of $\delta(\tuple{S,s},w_{[0..0)}) = \tuple{B,b}$ and $w_{[0..]} \models \beta$ does not hold, we only need to prove the right-hand side for $i > 0$.
 
Since we have $w \models \varphi$, then either $w$ satisfies the disjunct $\psi$ in line (2) or there exists a $\tuple{\sigma, T'} \in \Stay{s}$ with $\delta(\tuple{T',s}, \sigma) = \tuple{T,t}$ such that the matching disjunct $\chi$ in (1) is satisfied by $w$. We assume that the first case holds, and defer the second case to later.

Note, that when applying the induction hypothesis (I.H.2) to each conjunct of $\psi$, we can simplify the nested existential quantification to $\false$, since $w \not \models \false$. Thus we obtain for all $i \geq 0$:

\begin{enumerate}[label=(\alph*)]
	\item For every $\tuple{\eta, L} \in \Leave{s}$ at least one of the following statements holds:
	\begin{enumerate}[label=(\roman*)]
		\item $\delta(S, w_{[0..i)}) \neq L$
		\item $w[i] \neq \eta$
	\end{enumerate}
	\item For every $\tuple{\rho, B'} \in \Stay{s}$ such that $\delta(\tuple{B',s},\rho) = \tuple{B,b}$ at least one of the following statements holds:
	\begin{enumerate}[label=(\roman*)]
	    \item $\delta(S, w_{[0..i)}) \neq B'$
	    \item $w[i] \neq \rho$
	    \item $w_{[i+1..]} \not \models \beta$
	\end{enumerate}
\end{enumerate}

Because $w$ and the sequence of generated configurations does not contain a combined letter that leaves $s$ due to (a), we do have $\tuple{w[i], \delta(S, w_{[0..i)})} \notin \Leave{s}$ for all $i \geq 0$. Thus $\tuple{w[i], \delta(S, w_{[0..i)})} \in \Stay{s}$ for all $i \geq 0$.  Further, due to (b) we either do not see a combined letter at position $i$ that enters the configuration $\tuple{B,b}$ (and thus $\delta(\tuple{S,s}, w_{[i+1..]}) = \tuple{B,b}$) or the infinite suffix of $w_{[i+1..]}$ does not satisfy $\beta$. From this, we can conclude that the right-hand side holds. (The case for $i=0$ was already established in the second paragraph.)

We now assume that there exists a tuple $\tuple{\sigma, T'}$ such that $w$ satisfies the matching disjunct $\chi$ of line (1). Note that the existence of $\tuple{\sigma, T'}$ immediately implies $s=t$. Moreover, we can assume that $w$ does not satisfy the disjunct $\psi$ from line (2). Thus by applying (I.H.2), we know that there exists an $k \geq 0$ such that either $\delta(S, w_{[0..k)}) = L$ and $w[k] = \eta$ from $\tuple{\eta, L} \in \Leave{s}$ or $\delta(S, w_{[0..k)}) = B$, $w[k] = \rho$, and $w_{[k+1..]} \models \beta$ for some $\tuple{\rho, B'} \in \Stay{s}$ and $\delta(\tuple{B',s}, \rho) = \tuple{B,b}$. Thus when we apply the (I.H.2) to the corresponding reachability formula of $\chi$ we obtain a $k' \in [0..k)$ such that $\delta(S,w_{[0..k')}) = T'$ and $w_{[k'..]} \models \sigma \land \X \tau$. Observe that $\chi$ is a conjunction of formulas with the shape \[\WeakReach{S}{\cdot}{\cdot}{T'}{\sigma \land \X \tau}\] and we now apply the induction hypothesis to all reachability formulas. Since they all share the same ``target'', we can instantiate the nested existential quantification to the same $j'$ (where $j'$ is the smallest non-negative integer satisfying the condition) such that:

\begin{enumerate}[label=(\alph*)]
	\item $\delta(S, w_{[0..j')}) = T'$
	\item $w[j'] = \sigma$
	\item $w_{[j'+1..]} \models \tau$
	\item For every $\tuple{\eta, L} \in \Leave{s}$ and every $i' \in [0..j']$ at least one of the following statements holds:
	\begin{enumerate}[label=(\roman*)]
		\item $\delta(S, w_{[0..i')}) \neq L$
		\item $w[i'] \neq \eta$
	\end{enumerate}
	\item For every $\tuple{\rho, B'} \in \Stay{s}$ such that $\delta(\tuple{B',s},\rho) = \tuple{B,b}$  and every $i' \in [0..j']$ at least one of the following statements holds:
	\begin{enumerate}[label=(\roman*)]
	    \item $\delta(S, w_{[0..i')}) \neq B'$
	    \item $w[i'] \neq \rho$
	    \item $w_{[i'+1..]} \not \models \beta$
	\end{enumerate}
\end{enumerate}

In order to show that $w$ satisfy the right-hand side of the equation, let now $i \geq 0$ be an arbitrary integer and we need to prove:
\[\begin{array}{c}
\big((\delta(\tuple{S,s},w_{[0..i)}) = \tuple{B,b} \land w_{[i..]} \models \beta) \lor (i > 0 \land \tuple{w[i{-}1], \delta(S,w_{[0..i)})} \in \Leave{s})\big) \\ 
\rightarrow (\exists j \in [0..i). ~ \delta(\tuple{S,s},w_{[0..j)}) = \tuple{T,t} \land w_{[j..]} \models \tau)
\end{array}\]

The case $i = 0$ was already discussed in the second paragraph. Further, if $i > j' + 1$, then (a-c) show that $j$ can be instantiated with $j'+1$. Thus assume that $i \in [1..(j'+1)]$. Further, since $j'$ was chosen to be the smallest integer such that (a-c) holds, we can simplify the expression we need to prove to:
\[(\delta(\tuple{S,s},w_{[0..i)}) \neq \tuple{B,b} \lor w_{[i..]} \not \models \beta) \land \tuple{w[i{-}1], \delta(S,w_{[0..i{-}1)})} \in \Stay{s})\]

Note that second conjunct is direct consequence of (d). For the first conjunct, we proceed by contradiction and assume that $i$ is the smallest integer such that:
\[\delta(\tuple{S,s},w_{[0..i)}) = \tuple{B,b} \land w_{[i..]} \models \beta\]

Due to the already established second conjunct, we have $\tuple{w[i-1], \delta(S, w_{[0..i-1)}}) = \tuple{\rho, B'} \in \Stay{s}$ and $\delta(\tuple{B',s}, \rho) = \tuple{B,b}$. But since $i - 1 \in [0..j']$, we have a contradiction to (e).

\medskip
\noindent ($\Leftarrow$): We assume that $w$ satisfies the right-hand side.
We first consider the case that the nested existential quantification is not true for any $i \geq 0$. Thus for all $i \geq 0$ it holds that:
\begin{enumerate}[label=(\alph*)]
    \item $\tuple{w[i], \delta(S, w_{[0..i)})} \in \Stay{s}$ 
	\item $\delta(\tuple{S,s},w_{[0..i)}) \neq \tuple{B,b}$ or $w_{[i..]} \not \models \beta$
\end{enumerate}

Since $w$ satisfies the right-hand side, it is easy to see that we have either $\tuple{S,s} \neq \tuple{B,b}$ or $w \not \models \beta$. Hence it remains to show that $w \models \WeakStayReachNonEmpty{\tuple{S,s}}{\tuple{B,b}}{\beta}{\tuple{T,t}}{\tau}$, which we do by showing that $w$ satisfies (2). From (a,b) we obtain: 

\begin{enumerate}
    \item[(a')] For every $\tuple{\eta, L} \in \Leave{s}$ we have either $\delta(S,w_{[0..i)}) \neq L$ or $w[i] \neq \eta$.
	\item[(b')] Either $w_{[i+1..]} \not \models \beta$ or for every $\tuple{\rho, B'} \in \Stay{s}$ such that $\delta(\tuple{B',s},\rho) = \tuple{B,b}$ we have either $\delta(S,w_{[0..i)}) \neq B'$ or $w[i] \neq \rho$.
\end{enumerate}

We then apply the induction hypothesis (I.H.2) to (a',b') and obtain that $w$ satisfies the disjunct of line (2).

We now consider the second case and assume that there is a $j$ such that:

\begin{enumerate}[label=(\alph*)]
	\item $\delta(\tuple{S,s}, w_{[0..j)}) = \tuple{T,t}$
	\item $w_{[j..]} \models \tau$
\end{enumerate}

Without loss of generality, we can assume $j$ to be the smallest integer with such a property. Further, since $w$ satisfies the right-hand side, we have:

\begin{enumerate}
	\item[(c)] $\forall i \in [0..j). ~ \tuple{w[i], \delta(S, w_{[0..i)})} \in \Stay{s}$
	\item[(d)] $\forall i \in [0..j]. ~ \delta(\tuple{S,s},w_{[0..i)}) \neq \tuple{B,b} \lor w_{[i..]} \not \models \beta$ 
\end{enumerate}

If $j = 0$, then $w$ immediately satisfies the left-hand side. Thus assume $j > 0$. Due to (c), we have that $\tuple{\sigma, T'} = \tuple{w_{[j-1]}, \delta(S, w_{[0..j-1)})} \in \Stay{s}$ and $\delta(\tuple{\delta(S, w_{[0..j-1)}), s},w_{[j-1]}) = \tuple{T, t}$. Thus there exists a matching disjunct for $\tuple{\sigma, T'}$ in line (1) and we know it is satisfied by $w$ due the $(a,b,d)$ and the induction hypothesis (I.H.2).

\bigskip
\noindent \textbf{Reachability formula 5}

\noindent We want to prove the following equivalence:
\[\begin{array}{cl}
     & w \models \LeaveReach{\tuple{S,s}}{\tuple{B,b}}{\beta}{\tuple{T,t}}{\tau} \\[1.8em]
\iff & \exists i_1, i_2 \geq 0. ~ \delta(\tuple{S,s},w_{[0..i_1)}) = \tuple{T,t} \land w_{[i_1..]} \models \tau \\ 
     & \phantom{\exists i_1, i_2 \geq 0. ~} \land ~ (\exists j_1 \in [0..i_1). ~ \tuple{w[j_1], \delta(S, w_{[0..j_1))} \in \Enter{s}}) \\
     & \phantom{\exists i_1, i_2 \geq 0. ~} \land ~ \tuple{w[i_2], \delta(S, w_{[0..i_2)})} \in \Leave{s} \\ 
     & \phantom{\exists i_1, i_2 \geq 0. ~} \land ~ (\forall j_2 \in [0..\max(i_1{-}1, i_2)]. \\ & \qquad \qquad \qquad \qquad \delta(\tuple{S,s},w_{[0..j_2)}) \neq \tuple{B,b} \lor w_{[j_2..]} \not \models \beta)	
\end{array}\]

\noindent ($\Rightarrow$): We assume that \[w \models \LeaveReach{\tuple{S,s}}{\tuple{B,b}}{\beta}{\tuple{T,t}}{\tau}\] holds and proceed by first constructing a witness for $i_2$ and then one for $i_1$.

\smallskip
\noindent($\exists i_2$): 
Since $w$ satisfies line (3), there must be a combined letter $\tuple{\sigma,L} \in \Leave{s}$ such that $w$ satisfies the corresponding disjunct. We apply to this disjunct the induction hypothesis (I.H.3) and instantiate the existential quantifier to $i_2$ (which we intend to be the witness for $\exists i_2$) with the following properties:

\begin{enumerate}
	\item $\delta(\tuple{S,s}, w_{[0..i_2)}) = \tuple{L,s}$
	\item $w[i_2] = \sigma$ and $\delta(\tuple{S,s}, w_{[0..i_2)}) \neq \tuple{B,b} \lor w_{[i_2..]} \not \models \beta$
	\item $\forall j_2 \in [0..i_2). ~ \delta(\tuple{S,s}, w_{[0..j_2)}) \neq \tuple{B,b} \lor w_{[j_2..]} \not \models \beta$
	\item $\forall j_2 \in [0..i_2). ~ \tuple{w[j_2], \delta(S,w_{[0..j_2)})} \in \Stay{s}$.  
\end{enumerate}
Observe that due to (1,2) we have $\tuple{w[i_2], \delta(S, w_{[0..i_2)})} \in \Leave{s}$ and together with (2,3) one can see that $i_2$ is a sufficient witness for the $\exists i_2$ in the right-hand side of the equation we want to prove. Note that property (4) is not useless and will be of importance later.

\smallskip
\noindent($\exists i_1$):
Since $w$ also satisfies lines (1,2), there must be a combined letter $\tuple{\sigma, T'} \in \Enter{t}$ such that $w$ satisfies the corresponding disjunct. We then apply the induction hypothesis (I.H.1) to all terms of the matching conjunction in lines (1,2). Since all reachability formulas share the same target configuration and formula, we can instantiate all existential quantifiers to the same integer $k$ by picking the smallest instance for each existential quantifier. Let now $k$ be this smallest non-negative integer. We introduce the following two abbreviations $T'' = \delta(\tuple{T',\cdot}, \sigma)$ and $R'_\eta = \delta(\tuple{R,\cdot},\eta)$ for some $\tuple{\eta, R} \in \Enter{b}$. Note that we can leave out the last state in the definition of $T''$, since for all states $q,p$ the transition relation maps the the same successor configuration, i.e., $\delta(\tuple{T', q}, \sigma) = \delta(\tuple{T', p}, \sigma)$. Analogously, for we omit it from $R'_\eta$. We now list all relevant properties of $k$ derived from applying the induction hypothesis:
\begin{enumerate}
	\item[5.] $\delta(S, w_{[0..k)}) = T'$
	\item[6.] $w[k] = \sigma$
	\item[7.] $w_{[k+1..]} \models \StayReach{T''}{\tuple{B,b}}{\beta}{\tuple{T,t}}{\tau}$
	\item[8.] For all $\tuple{\eta, R} \in \Enter{b}$ and all $\ell \in [0..k)$ at least one of the following statements is true:
	\begin{enumerate}
		\item $\delta(S, w_{[0..\ell)}) \neq R$ 
		\item $w[\ell] \neq \eta$
		\item $w_{[\ell+1..]} \not \models \WeakStayReach{R'_\eta}{\tuple{T,t}}{\tau}{\tuple{B,b}}{\beta}$
	\end{enumerate}
\end{enumerate}
\noindent We continue and apply the induction hypothesis to (7) and then obtain $i_1$ (which we intend to be the witness for $\exists i_1$) with the following properties (already shifted to indices relative to $w$). Further, we can assume $i_1$ to be the smallest integer with these properties.

\begin{enumerate}
	\item[9.] $\delta(T'', w_{[k+1..i_1)}) = \tuple{T,t}$ and $w_{[i_1..]} \models \tau$
	\item[10.] $\forall j \in [(k+1).. i_1). ~ \delta(T'', w_{[k+1..j)}) \neq \tuple{B,b} \lor w_{[j..]} \not \models \beta$
	\item[11.] $\forall j \in [(k+1)..i_1). ~ \tuple{w[j], \delta(\delta(T',\sigma), w_{[k+1..j)})} \in \Stay{t}$
\end{enumerate}

\noindent
From (5,6,9) we obtain that $\delta(\tuple{S,s}, w_{[0..i_1)}) = \tuple{T,t}$ and that $w_{[i_1..]} \models \tau$ taking care of the first line of the right-hand side of the equation. Further, by construction $\tuple{w[k], \delta(S, w_{[0..k)})} = \tuple{\sigma, T} \in \Enter{t}$ and since $k < i_1$, we now know that the second line is satisfied by our witness for $\exists i_1$. Next, we inline the definition of $T''$ in (10) and obtain: $\forall j \in [(k+1)..i_1). ~ \delta(\tuple{S,s}, w_{[0..j)}) \neq \tuple{B,b} \lor w_{[j..]} \not \models \beta$. Thus in order to complete this direction of the proof it remains to show that:
\[
\forall j \in [(i_2+1)..k]. ~ \delta(\tuple{S,s}, w_{[0..j)}) \neq \tuple{B,b} \lor w_{[j..]} \not \models \beta
\]
If $[(i_2+1)..k]$ is the empty set, we are done. Thus assume that $i_2 < k$. We proceed with a proof by contradiction and assume that there exists an index $j \in [i_2+1..k]$ such that $\delta(\tuple{S,s}, w_{[0..j)}) = \tuple{B,b}$ and $w_{[j..]} \models \beta$. Due to (1) and (4), there must be some $j' \in [i_2.. j{-}1]$ such that the combined letter $\tuple{\eta, R} = \tuple{w[j'], \delta(S, w_{[0..j')})}$ is in $\Enter{b}$. Without loss of generality we can further assume that:
\begin{enumerate}
	\item [12.] $\forall j'' \in [j'..j).~\tuple{w[j''], \delta(S, w_{[0..j'')})} \in \Stay{b}$
\end{enumerate}
Note that $j' < k$ and since (8a) and (8b) cannot be true we can instantiate (8c) to:
\[w_{[j'+1..]} \not \models \WeakStayReach{R'_\eta}{\tuple{T,t}}{\tau}{\tuple{B,b}}{\beta}\]
\noindent We now apply the induction hypothesis to this and obtain an index $\ell \geq j' + 1$ such that the following two statements hold (already with adjusted indices):

\begin{enumerate}
	\item[13.] At least one of the following statements hold:
	\begin{enumerate} 
		\item $\delta(\tuple{S,s}, w_{[0..\ell)}) = \tuple{T,t} \wedge w_{[\ell..]} \models \tau$ 
		\item $\ell > j'+1 \wedge \tuple{w[\ell-1], \delta(\tuple{S,s}, w_{[0..\ell-1)}} \in \Leave{b}$
	\end{enumerate}
	\item[14.] $\forall j'' \in [(j'{+}1)..\ell). ~ \delta(\tuple{S,s}, w_{[0..j'')}) \neq \tuple{B,b} \lor w_{[j''..]} \not \models \beta$
\end{enumerate}

Since $j' < j$ and since we assumed that $j$ is the smallest integer satisfying $\delta(\tuple{S,s}, w_{[0..j)}) = \tuple{B,b}$ and $w_{[j''..]} \models \beta$, we can conclude with (14) that $\ell \leq j$. However, we have a contradiction between (13a), $\ell \leq j \leq k < i_1$, and $i_1$ being the smallest index reaching $\tuple{T,t}$ and satisfying $\tau$. Thus (13b) must hold. However, since $\ell \leq j$ this contradicts (12). Thus there cannot be such a $j$ and we can conclude this direction.

\medskip
\noindent ($\Leftarrow$): Assume that $w$ satisfies the left-hand side. Then there exists $i_1, i_2, j_1 \geq 0$, further let $i_1, i_2$ be the smallest integers and let $j_1$ be the largest integer such that:

\begin{enumerate}[label=(\alph*)]
	\item $\delta(\tuple{S,s},w_{[0..i_1)}) = \tuple{T,t}$
	\item $w_{[i_1..]} \models \tau$
	\item $j_1 < i_1$
	\item $\tuple{w[j_1], \delta(S, w_{[0..j_1)})} \in \Enter{t}$
	\item $\tuple{w[i_2], \delta(S, w_{[0..i_2)})} \in \Leave{s}$
	\item $\forall j_2 \in [0..\max(i_1-1,i_2)]. ~ \delta(\tuple{S,s},w_{[0..j_2)}) \neq \tuple{B,b} \lor w_{[j_2..]} \not \models \beta$
\end{enumerate}

Since $i_2$ is the smallest integer with this property, we have $\tuple{w[k], \delta(S, w_{[0..k)})} \in \Stay{s}$ for all $k \in [0..i_2{-}1]$. Thus we can apply (I.H.3) to this and (e,f) to show that $w$ satisfies line (3) for $\tuple{w[i_2], \delta(S, w_{[0..i_2)})} = \tuple{\sigma, L} \in \Leave{s}$. In order to show the disjunct of (1) and (2), we need to identify a suitable $\tuple{\sigma, T'} \in \Enter{t}$. We argue that $\tuple{\sigma, T'} = \tuple{w[j_1], \delta(S, w_{[0..j_1)})}$ is a suitable choice.

We proceed by first showing that $w$ satisfies line (1). Since $j_1$ is the largest integer with such a property, we immediately get $\tuple{w[k], \delta(S, w_{[0..k)})} \in \Stay{s}$ for all $k \in [j_1{+}1,i_1{-}1]$. By applying (I.H.3) we get:

\[w_{[j_1+1..]} \models \StayReach{\delta(\tuple{T', \cdot}, \sigma)}{\tuple{B,b}}{\beta}{\tuple{T,t}}{\tau}\]

We now apply (I.H.1) and use (d) to establish that the first line is satisfied by $w$. In order show that line (2) is satisfied it remains to show that for all $k \in [0..j_1]$ and all $\tuple{\eta, L} \in \Enter{b}$ either $\delta(\tuple{S,s}, w_{[0..k)}) \neq \tuple{\eta, L}$ or $w_{[k..]} \not \models \eta \land \X \Big( \StayReachMulti{\delta(\tuple{R,\cdot},\eta)}{\tuple{T,t}}{\tau}{\tuple{B,b}}{\beta}{\Weak} \Big)$.

We proceed by contradiction. Assume that there is $k \leq j_1$ such that $\delta(\tuple{S,s},$ $w_{[0..k)}) = \tuple{\eta, L}$, $w[k] = \eta$, and $w_{[k+1..]} \models \WeakStayReach{\delta(\tuple{R,\cdot}, \eta)}{\tuple{T,t}}{\tau}{\tuple{B,b}}{\beta})$. By applying (I.H.4) and adjusting the indices we obtain:

\[\begin{array}{rl}
	\forall i \geq k+1. & \big((\delta(\tuple{S,s},w_{[0..i)}) = \tuple{T,t} \land w_{[i..]} \models \tau) \\
                               & ~ ~ ~ \lor ~ (i > k + 1 \land \tuple{w[i{-}1], \delta(S, w_{[0..i{-}1)})} \in \Leave{b}) \big) \\ 
  	                           & \rightarrow ~ (\exists j \in [k+1..i). ~ \delta(\tuple{S,s},w_{[0..j)}) = \tuple{B,b} \land w_{[j..]} \models \beta)
\end{array}\]

By instantiating this with (a,b) we get:

\[
\exists j \in [k+1..i_1). ~ \delta(\tuple{S,s},w_{[0..j)}) = \tuple{B,b} \land w_{[j..]} \models \beta
\]

However this contradicts (f).

Thus we can apply (I.H.1) to establish that $w$ satisfies (2).
\end{proof}

\lemSyntacticSafe*

\begin{proof}
Observe first that there is no circularity in the definitions of the five reachability formulas, even though they are defined by each other: Formula 2 is defined on top of formula 1, which is defined on top of formulas 3 and 5, while formulas 3, 4, and 5 are defined with respect to reachability formulas over configurations of a lower level.

Let $i \geq 1$ be an arbitrary index. We prove the statement by induction on the level $m$ of the configurations $S$, $B$, $T$ in the reachability formulas. We split the proof to five cases corresponding to the five reachability formulas.

\medskip
\noindent \textbf{Reachability formula 1.}
Let $\beta \in \Pi_i$ and let $\tau \in \Sigma_i$. We proceed by an induction on the the configuration level $m$. \smallskip

\noindent ($m = 0$): There is only one configuration of level $0$ and it is the empty configuration. Thus $S = T = B = \tuple{}$. Applying the definition we obtain:
\[\Reach{\tuple{}}{\tuple{}}{\beta}{\tuple{}}{\tau} = (\neg \beta) \U \tau\]
Note that $\neg \beta \in \Sigma_i$ and thus the whole formula $(\neg \beta) \U \tau$ is also in $\Sigma_i$. \smallskip

\noindent ($m \rightarrow m + 1$): Let $\tuple{S,s}, \tuple{T,t}, \tuple{B,b} \in Q_1 \times \cdots \times Q_{m+1}$ be arbitrary configurations of level $m+1$. Thus:
\[\Reach{\tuple{S,s}}{\tuple{B,b}}{\beta}{\tuple{T,t}}{\tau} = \StayReach{\tuple{S,s}}{\tuple{B,b}}{\beta}{\tuple{T,t}}{\tau} \lor \LeaveReach{\tuple{S,s}}{\tuple{B,b}}{\beta}{\tuple{T,t}}{\tau}\]
 We now can use apply the induction hypothesis and obtain that the formula is in $\Sigma_i$. 

\medskip
\noindent \textbf{Reachability formula 2.}
Let $\beta \in \Sigma_i$ and let $\tau \in \Pi_i$. Then by induction hypothesis $\Reach{S}{T}{\tau}{B}{\beta}$ is in $\Sigma_i$ and thus $\WeakReach{S}{B}{\beta}{T}{\tau}$ is in $\Pi_i$.

\medskip
\noindent \textbf{Reachability formula 3.}
Let $\beta \in \Pi_i$ and let $\tau \in \Sigma_i$. Note that under this assumption $\false$, $\eta$, and $\rho \land \X \beta$ belong to $\Pi_i$ and that $\sigma \land \X \tau$ belongs to $\Sigma_i$. Thus by applying the induction hypothesis we obtain that $\StayReachNonEmpty{\tuple{S,s}}{\tuple{B,b}}{\beta}{\tuple{T,t}}{\tau}$ is in $\Sigma_i$. Since $\neg \beta$ also belongs to $\Sigma_i$, we obtain that $\StayReach{\tuple{S,s}}{\tuple{B,b}}{\beta}{\tuple{T,t}}{\tau}$ is from $\Sigma_i$.

\medskip
\noindent \textbf{Reachability formula 4.}
Let $\beta \in \Sigma_i$ and let $\tau \in \Pi_i$. Note that under this assumption $\eta$ and $\rho \land \X \beta$ belong to $\Sigma_i$ and that $\false$, $\neg \beta$, and $\sigma \land \X \tau$ belong to $\Pi_i$. We then apply again the induction hypothesis and obtain that $\WeakStayReachNonEmpty{\tuple{S,s}}{\tuple{B,b}}{\beta}{\tuple{T,t}}{\tau}$ and thus also $\WeakStayReach{\tuple{S,s}}{\tuple{B,b}}{\beta}{\tuple{T,t}}{\tau}$ belongs to $\Pi_i$.

\medskip
\noindent \textbf{Reachability formula 5.}
Let $\beta \in \Pi_i$ and let $\tau \in \Sigma_i$. By an analogous argument to the two preceding cases we obtain that the formulas in line (1) and (3) belong to $\Sigma_i$. Further, note that by the induction hypothesis $\eta \land \X \Big( \StayReachMulti{\delta(\tuple{R,\cdot},\eta)}{\tuple{T,t}}{\tau}{\tuple{B,b}}{\beta}{\Weak} \Big)$ belongs to $\Pi_i$ and thus the overall reachability formula 5 belongs to $\Sigma_i$.
\end{proof}

\subsection{Proofs from~\cref{sec:SizeAnalysis}}
\label{app:length}

\ReachFormulaSize*

\begin{proof}
We first prove the upper bound on the depth of the formula and then move on to the claim about the length of the formula.

\noindent \textbf{Depth Analysis.}
We prove the claim by induction on the level $i$.
For the base case of $i=0$, the main reachability formula is just  $(\neg \beta) \U \tau$, having $\DF(0,d) = d+1$, which is equal to $d+3^0$, as required.
For the induction step of a general level $i>0$, we will first establish a bound on $\DF(i,d)$ that is relative to $\DF(i{-}1,d)$, and then get from it an absolute bound, using the induction hypothesis.

Observe that $\DF(i,d)$ is bounded in \cref{tab:size} to the maximum between $\DF_3(i,d)$ and $\DF_5(i,d)$, while the bound on $\DF_5(i,d)$ is at least as on $\DF_3(i,d)$. Hence:
\begin{align*}
\DF(i,d) ~\leq~ \DF_5(i,d)~ & \leq &  \DF(i{-}1, \max(1 + \DF_3(i,d), 1 + \DF_4(i,d) )) \\
& \leq &  \DF(i{-}1, 1 + \DF(i{-}1,d+1))
\end{align*}
Applying the induction hypothesis on $\DF(i{-}1,d)$, we get an absolute bound:
\begin{align*}
	\DF(i,d) ~& \leq &  \DF(i{-}1, 1 + (d +1 +3^{i{-}1}))~~~~~~~~~~~~~~~~~~~~~~~~ \\
					& \leq &   2 + d +3^{i{-}1}+3^{i{-}1} ~~~ \leq  ~~~~~~~~~ d +3^i 
\end{align*}

\medskip
\noindent \textbf{Length Analysis.}
We prove the claim by induction on the level $i$.
For the base case of $i=0$, the main reachability formula is just  $(\neg \beta) \U \tau$, having $\LF(0,l) = 2+2l$, which is indeed not larger than $l \cdot  (10|\Sigma|^2n)^{4^i}$.
For the induction step of a general level $i>0$, we will first establish a bound on $\LF(i,l)$ relative to $\LF(i{-}1,\cdot)$, and then apply the induction hypothesis to get an absolute bound.

Observe that $\LF(i,l)$ is bounded in \cref{tab:size} with respect to $\LF_3(i,\cdot)$ and $\LF_5(i,\cdot)$, and
$\LF_3(i,l)$ is already bounded in \cref{tab:size} with respect to $\LF(i{-}1,\cdot)$. 
We simplify the $\LF_5$ bound, and substitute the bound on $\LF_3$, getting:

\begin{align*}
	\LF_5(i,l)& \leq & |\Sigma|n^{i-1} \cdot \big(\LF_1(i-1,3+\LF_3(i,l))+2+\\
		&& |\Sigma| n^{i-1}\cdot \big(\LF_1(i-1, \max(3+\LF_3(i,l),3+\LF_4(i,l)))+1\big)\big) +\\
		&& 1+ |\Sigma|n^{i-1}\cdot (1+\LF_3(i,3+l)) \\
		& \leq & |\Sigma|n^{i-1} \cdot \big(\LF_1(i-1,3+\LF_3(i,l))+2+\\
		&& |\Sigma| n^{i-1}\cdot \big(\LF_1(i-1, 3+\LF_3(i,l))+1\big)\big) +\\
		&& 1+ |\Sigma|n^{i-1}\cdot (1+\LF_3(i,3+l)) \\
		& \leq & |\Sigma|^2n^{2(i-1)}\LF_1(i-1,3+\LF_3(i,l)) + \\
		&& |\Sigma|n^{(i-1)}\LF_1(i-1,3+\LF_3(i,l)) + \\
		&& |\Sigma|n^{(i-1)}\LF_3(i,3+l)) + \\
		&& |\Sigma|^2n^{2(i-1)} + \\
		&& 3|\Sigma|n^{(i-1)} + 1\\
		& \leq & 4|\Sigma|^2n^{2(i-1)}\LF_1(i-1,3+\LF_3(i,l+3)) \\
		& \leq & 4|\Sigma|^2n^{2(i-1)}\LF_1(i-1,12+ 2l+ 4|\Sigma|^2n^{2(i-1)}\LF_1(i{-}1,l{+}3))
\end{align*}
We can now bound $\LF(i,l)$ relative to $\LF(i{-}1,\cdot)$, getting:
\begin{align*}
	\LF(i,l)&\leq &  1+	\LF_3(i,l) +\LF_5(i,l)\\
	&\leq & 	4+ 2l+ 4|\Sigma|n^{2(i-1)}\LF_1(i{-}1,l{+}3) +\\
	&&4|\Sigma|^2n^{2(i-1)}\LF_1(i-1,12+ 2l+ 4|\Sigma|^2n^{2(i-1)}\LF_1(i{-}1,l{+}3))\\
	& \leq & 5|\Sigma|^2n^{2(i-1)}\LF_1(i-1,12+ 2l+ 4|\Sigma|^2n^{2(i-1)}\LF_1(i{-}1,l{+}3))\\
	& \leq & 5|\Sigma|^2n^{2(i-1)}\LF_1(i-1,5|\Sigma|^2n^{2(i-1)}\LF_1(i{-}1,l{+}3))
\end{align*}
Applying the induction hypothesis on $\LF(i{-}1,\cdot)$, we get an absolute bound:
\begin{align*}
\LF(i,l)& \leq & 5|\Sigma|^2n^{2(i-1)}\LF_1(i-1,5|\Sigma|^2n^{2(i-1)}\LF_1(i{-}1,l{+}3))\\
& \leq & 5|\Sigma|^2n^{2(i-1)}(5|\Sigma|^2n^{2(i-1)}(l+3)(2|\Sigma|n)^{4^{i-1}})(10|\Sigma|^2n)^{4^{i-1}}\\
& \leq & 25|\Sigma|^4n^{4(i-1)}(l+3)(10|\Sigma|^2n)^{2\cdot 4^{i-1}}\\
&\leq & l(10|\Sigma|^2 n)^{4^i}
\end{align*}
The last inequality clearly holds for $i\geq 2$ as the $(10|\Sigma|^2n)^{4^{i}}$ term becomes very large.
For $i=1$, it also holds as the $n^{4(i-1)}$ term disappears and we get $25\cdot  10^2(l+3) \leq 10^4 l n^2$, which holds even for $l=n=1$.
\end{proof}

\subsection{Proofs from~\cref{sec:DetAutomataToLTL}}
\label{app:main}

\LTLforVisitFinitelyOften*

\begin{proof}
	Observe that $C$ is visited finitely often iff either $C$ is not visited at all, which  by \cref{thm:reach-formulas-correctness} is formulated by \smallerReach{$\neg(\ReachTT{\InitState}{C})$} or there exists a last visit to $C$, namely a visit to $C$ after which there is no other visit to $C$, which is formulated by \smallerReach{$\ReachTT{\InitState}{C}(\neg (\NonEmptyReachTT{C}{C}))$}. Note that by \Cref{lem:reach-formulas-syntactic} we have \smallerReach{$\ReachTT{\InitState}{C}, \NonEmptyReachTT{C}{C} \in \Sigma_1$}. Thus \smallerReach{$\neg \ReachTT{\InitState}{C}, \neg \NonEmptyReachTT{C}{C} \in \Pi_1$} and then again by \Cref{lem:reach-formulas-syntactic} we get \smallerReach{$\ReachTT{\InitState}{C}(\neg (\NonEmptyReachTT{C}{C})) \in \Sigma_2$}. Thus $\FIN(C) \in \Sigma_2$.
\end{proof}

\AutomatonToLTL*

\begin{proof}
We complete the proof of the Theorem and give a complete analysis of all six acceptance conditions:

\begin{itemize}
    \item $\D$ is a Muller automaton: the overall formula $\varphi$ is in $\Delta_2$, since it is a Boolean combination of $\FIN(C)$ formulas, which by \cref{cl:LTLforVisitFinitelyOften} belong to $\Sigma_2$.
	\item $\D$ is a coB\"uchi automaton: we construct the formula $\varphi$ directly from the  coB\"uchi condition $\alpha$, having a conjunction of $\FIN(C)$ formulas, over all configurations $C$ that are mapped to states in $\alpha$. As $\FIN(C)$ belongs to $\Sigma_2$, so does $\varphi$.
	\item $\D$ is a B\"uchi automaton: we can complement it, apply the above argument over the resulting coB\"uchi automaton, and negate the resulting formula to obtain a formula from $\Pi_2$.
	\item $\D$ is a looping-coB\"uchi automaton: Let $s\in Q$ be the unique sink state that all accepting runs end up in, and let $S$ be the set of configurations mapped to $s$. We then define $\varphi$ as \smallerReach{$\bigvee_{C \in S} \ReachTT{\InitState}{C}$}. Note that $\phi$ belongs to $\Sigma_1$ since every disjunct belongs to $\Sigma_1$ due to \cref{lem:reach-formulas-syntactic}.
	\item $\D$ is a looping-B\"uchi automaton: the dual of the previous argument.
	\item $\D$ is a weak automaton: Let $G \subseteq Q$ be an accepting SCC of $\D$ and $G' \subseteq Q$ be all states that are reachable from $G$, but are not in $G$. Let $H$ and $H'$ be the set of configurations that are mapped to $G$ and $G'$, respectively. Then by \cref{thm:reach-formulas-correctness}, we have that $(\bigvee_{C \in H} \ReachTT{\InitState}{C}) \land (\bigwedge_{C' \in H'} \neg \ReachTT{\InitState}{C'})$ exactly captures all words that are accepted by eventually ending up in $G$. The overall formula $\varphi$ is then the disjunction of these formulas constructed for each accepting SCC of $\D$. The membership in $\Delta_1$ then follows immediately from \cref{lem:reach-formulas-syntactic}. \qedhere
\end{itemize}
\end{proof}

\end{document}